\documentclass[11pt]{article}

\usepackage[letterpaper,margin=0.90in]{geometry}
\usepackage{amsmath, amssymb, amsthm, thmtools, amsfonts}
\usepackage{bbm}

\usepackage{ifthen}

\usepackage{tikz}
\usetikzlibrary{positioning,decorations.pathreplacing}

\usepackage{cite}
\usepackage{appendix}
\usepackage{graphicx}
\usepackage{color}
\usepackage{algorithm}
\usepackage[noend]{algpseudocode}
\usepackage{epstopdf}
\usepackage[textsize=tiny]{todonotes}

\newcommand{\matodo}[1]{\todo[color=red!20]{#1}}

\usepackage{framed}
\usepackage[framemethod=tikz]{mdframed}
\usepackage[bottom]{footmisc}
\usepackage[shortlabels]{enumitem}
\setitemize{noitemsep,topsep=3pt,parsep=3pt,partopsep=3pt}
\usepackage[font=small]{caption}
\usepackage{xspace}

\newtheorem{theorem}{Theorem}[section]
\newtheorem{lemma}[theorem]{Lemma}
\newtheorem{meta-theorem}[theorem]{Meta-Theorem}

\newtheorem{corollary}[theorem]{Corollary}

\newtheorem{observation}[theorem]{Observation}
\newtheorem{definition}[theorem]{Definition}

\definecolor{darkgreen}{rgb}{0,0.5,0}
\usepackage{hyperref}
\hypersetup{
    unicode=false,          
    colorlinks=true,        
    linkcolor=red,          
    citecolor=darkgreen,        
    filecolor=magenta,      
    urlcolor=cyan           
}
\usepackage[capitalize, nameinlink]{cleveref}
\crefname{theorem}{Theorem}{Theorems}
\Crefname{lemma}{Lemma}{Lemmas}
\Crefname{observation}{Observation}{Observations}
\algnewcommand\algorithmicswitch{\textbf{switch}}
\algnewcommand\algorithmiccase{\textbf{case}}

\algdef{SE}[SWITCH]{Switch}{EndSwitch}[1]{\algorithmicswitch\ #1\ \algorithmicdo}{\algorithmicend\ \algorithmicswitch}%
\algdef{SE}[CASE]{Case}{EndCase}[1]{\algorithmiccase\ #1}{\algorithmicend\ \algorithmiccase}%
\algtext*{EndSwitch}%
\algtext*{EndCase}%

\newcommand{\eps}{\varepsilon}

\newcommand{\local}{$\mathsf{LOCAL}$\xspace}

\newcommand{\poly}{\operatorname{\text{{\rm poly}}}}

\newcommand{\ceil}[1]{\lceil #1 \rceil}

\newcommand{\paren}[1]{\mathopen{}\left(#1\right)\mathclose{}}

\renewcommand{\paragraph}[1]{\vspace{0.15cm}\noindent {\bf #1}:}

\newcommand{\FullOrShort}{full}

\ifthenelse{\equal{\FullOrShort}{full}}{
	
  \newcommand{\fullOnly}[1]{#1}
  \newcommand{\shortOnly}[1]{}
  
  }{

    \newcommand{\fullOnly}[1]{}
    \newcommand{\IncludePictures}[1]{}
   
  }

\begin{document}

\date{}

\title{Sublogarithmic Distributed Algorithms for Lov{\'a}sz Local Lemma, 
\\and the Complexity Hierarchy}

\author{
	 Manuela Fischer\\
  \small ETH Zurich \\
  \small manuela.fischer@inf.ethz.ch
		\and
 Mohsen Ghaffari\\
  \small ETH Zurich \\
  \small ghaffari@inf.ethz.ch
 }

\maketitle

\setcounter{page}{0}
\thispagestyle{empty}

\begin{abstract}
Locally Checkable Labeling (LCL) problems include essentially all the classic problems of $\mathsf{LOCAL}$ distributed algorithms. In a recent enlightening revelation, Chang and Pettie [arXiv 1704.06297] showed that \emph{any} LCL (on bounded degree graphs) that has an $o(\log n)$-round randomized algorithm can be solved in $T_{LLL}(n)$ rounds, which is the randomized complexity of solving (a relaxed variant of) the Lov\'{a}sz Local Lemma (LLL) on bounded degree $n$-node graphs. Currently, the best known upper bound on $T_{LLL}(n)$ is $O(\log n)$, by Chung, Pettie, and Su [PODC'14], while the best known lower bound is $\Omega(\log\log n)$, by Brandt et al. [STOC'16]. Chang and Pettie conjectured that there should be an $O(\log\log n)$-round algorithm.

\medskip

Making the first step of progress towards this conjecture, and providing a significant improvement on the algorithm of Chung et al. [PODC'14], we prove that $T_{LLL}(n)= 2^{O(\sqrt{\log\log n})}$. Thus, any $o(\log n)$-round randomized distributed algorithm for any LCL problem on bounded degree graphs can be automatically sped up to run in $2^{O(\sqrt{\log\log n})}$ rounds. 

\medskip 
Using this improvement and a number of other ideas, we also improve the complexity of a number of graph coloring problems (in arbitrary degree graphs) from the $O(\log n)$-round results of Chung, Pettie and Su [PODC'14] to $2^{O(\sqrt{\log\log n})}$. These problems include defective coloring, frugal coloring, and list vertex-coloring.
\medskip

\end{abstract}

\newpage

\section{Introduction and Related Work}
The Lov\'{a}sz Local Lemma (LLL), introduced by Erd\H{o}s and Lov\'{a}sz in 1975\cite{LLL}, is a beautiful result which shows that, for a set of ``bad events'' in a probability space that have certain sparse dependencies, there is a non-zero probability that none of them happens. This result has become a central tool in the \emph{probabilistic method}\cite{Alon-Spencer}, when proving that certain combinatorial objects exist. Although the LLL itself does not provide an efficient way for finding these objects, and that remained open for about 15 years, a number of efficient centralized algorithms have been developed for it, starting with Beck's breakthrough in 1991\cite{beck1991LLL}, through \cite{alon1991LLL, molloy1998LLL, czumaj2000new, srinivasan2008improved, moser2009constructive}, and leading to the elegant algorithm of Moser and Tardos in 2010\cite{moser2010constructive}. See also \cite{chandrasekaran2010deterministic, kolipaka2011moser, Harris:2014:CAL:2634074.2634142, harris2013moser, Harris:2016, harris2016algorithmic} for some of the related work on that track.

In contrast, distributed algorithms for LLL and the related complexity are less well-understood. This question has gained an extraordinary significance recently, due to revelations that show that LLL is a ``complete'' problem for sublogarithmic-time problems. Next, we first overview the concrete statement of the LLL, and then discuss what is known about its distributed complexity, and what its special significance is for distributed algorithms. Then we proceed to presenting our contributions. \medskip

\subsection{The LLL and its Special Role in Distributed Algorithms}
\vspace{-5pt}
\paragraph{The Lov\'{a}sz Local Lemma} Consider a set $\mathcal{V}$ of independent random variables, and a family $\mathcal{X}$ of $n$ (bad) events on these variables. Each event $A \in \mathcal{X}$ depends on some subset $\text{vbl}(A)\subseteq \mathcal{V}$ of variables. Define the dependency graph $G_{\mathcal{X}}= (\mathcal{X}, \{ (A, B)\;|\; vbl(A)\cap vbl(B)\neq \emptyset\})$ that connects any two events which share at least one variable. Let $d$ be the maximum degree in this graph, i.e., each event $A \in \mathcal{X}$ shares variables with at most $d$ other events $B\in \mathcal{X}$. Finally, define $p=\max_{A\in \mathcal{X}} \Pr[A]$. The Lov\'{a}sz Local Lemma\cite{LLL} shows that $\Pr[\cap_{A\in \mathcal{X}} \bar{A}] >0$, under the \emph{LLL criterion} that $epd \leq 1$. Intuitively, if a local union bound is satisfied around each node in $G_{\mathcal{X}}$, with some slack, then there is a positive probability to avoid all bad events.



\medskip
\noindent\textbf{What's Known about Distributed LLL?} In the standard distributed formulation of LLL, we consider $\mathsf{LOCAL}$-model\cite{linial1987LOCAL, peleg:2000} algorithms that work on the $n$-node dependency graph $G_{\mathcal{X}}$, where per round each node can send a message to each of its neighbors\footnote{One can imagine a few alternative graph formulations, all of which turn out to be essentially equivalent in the $\mathsf{LOCAL}$ model, up to an $O(1)$ overhead in complexity.}.

Moser and Tardos\cite{moser2010constructive} provide an $O(\log^2 n)$-round randomized distributed algorithm. Chung, Pettie, and Su\cite{chung2014LLL} presented an $O(\log n \cdot \log^2 d)$-round algorithm, which was later improved slightly to $O(\log n \cdot \log d)$ \cite{Ghaffari-MIS}. Perhaps more importantly, under a modestly stronger criterion that $epd^2 < 1$, which is satisfied in most of the standard applications, they gave an $O(\log n)$-round algorithm\cite{chung2014LLL}. This remains the best known distributed algorithm. On the other hand, Brandt et al.\cite{brandt2016lower} showed a lower bound of $\Omega(\log \log n)$ rounds, which holds even if a much stronger LLL criterion of $p2^d<1$ is satisfied. Even under this exponentially stronger criterion, the best known upper bound changed only slightly to $O(\log n/ \log\log n)$\cite{chung2014LLL}. 

\medskip
\paragraph{Completeness of LLL for Sublogarithmic Distributed Algorithms} Very recently, Chang and Pettie\cite{chang2017time} showed that any $o(\log n)$-round randomized algorithm $\mathcal{A}$ for any Locally Checkable Labeling (LCL) problem $\mathcal{P}$---a problem whose solution can be checked in $O(1)$ rounds\cite{naor1995localicty}, which includes all the classic local problems---on bounded degree graphs can be transformed to an algorithm with complexity $O(T_{LLL}(n))$. Here, $T_{LLL}(n)$ denotes the randomized complexity for solving LLL on $n$-node bounded-degree graphs, with high probability. 

In a nutshell, their idea is to ``lie'' to the algorithm $\mathcal{A}$ and say that the network size is some much smaller value $n^* \ll n$. This deceived algorithm $\mathcal{A}$ may have a substantial probability to fail, creating an output that violates the requirements of the LCL problem $\mathcal{P}$ somewhere. However, the probability of failure in each local neighborhood is at most $1/n^{*}$. Choosing $n^{*}$ a large enough constant, depending on the complexity of $\mathcal{A}$, the algorithm $\mathcal{A}$ provides an LLL system---where we have one bad event for violation of each local requirement of $\mathcal{P}$--- that satisfies the criterion $pd^c<1$ for some (desirably large) constant $c\geq 1$. Hence, we can solve this LLL system and thus obtain a solution for the original LCL problem $\mathcal{P}$ in $O(T_{LLL}(n))$ time. 

This result implies that LLL is important not only for a few special problems, but in fact for essentially \emph{all} sublogarithmic-time distributed problems. Due to this remarkable role, Chang and Pettie state that ``\emph{understanding the distributed complexity of the LLL is a significant open problem.}'' Furthermore, although a wide gap between the best upper bound $O(\log n)$\cite{chung2014LLL} and lower bound $\Omega(\log\log n)$\cite{brandt2016lower} persists, they conjecture the latter to be tight:

\medskip
\noindent\textbf{Conjecture [Chang, Pettie \cite{chang2017time}]} There exists a sufficiently large constant $c$ such that the distributed LLL problem can be solved in $O(\log\log n)$ time on bounded degree graphs, under the symmetric LLL criterion $pd^c<1$.\footnote{This statement, as is, has a small imprecision: one should assume either that $d\geq 2$, in which case $pd^c<1$ can be replaced with $p(ed)^{c'}<1$ for some other constant $c'$, or that $pd^c<1/2$. Otherwise, two events of head or tail for a fair coin have $p=1/2$ and $d=1$, thus $pd^c< 1$, but one cannot avoid both.}


\subsection{Our Contributions}
\vspace{-7pt}
\paragraph{Faster Distributed LLL}
We make a significant step of progress towards this conjecture:  

\begin{restatable}{theorem}{RandLLLL}
\label{RANDLLL}
There is a $2^{O(\sqrt{\log\log n})}$-round randomized distributed algorithm that, with high probability\footnote{As standard, the phrase \emph{with high probability} (w.h.p.) indicates that an event has probability at least $1-n^{-c}$, for a sufficiently large constant $c$.}, solves the LLL problem with degree at most $d=O(\log^{1/5}\log n)$, under a symmetric polynomial LLL criterion $p(ed)^{32}<1$.\footnote{We remark that we did not try to optimize the constants.}
\end{restatable}
This improves over the $O(\log n)$-round algorithm of Chung et al.\cite{chung2014LLL}. We note that even under a significantly stronger exponential LLL criterion --- formally requiring $4ep 2^d d^4<1$ --- the best known round complexity was $O(\log n/\log\log n)$ \cite{chung2014LLL}. 
Furthermore, we note that a key ingredient in developing \Cref{RANDLLL} is a deterministic distributed algorithm for LLL, which we present in \Cref{DETLLL}. To the best of our knowledge, this is the first (non-trivial) deterministic distributed LLL algorithm. In fact, we believe that any conceivable future improvements on \Cref{RANDLLL} may have to improve on this deterministic part. 


Moreover, our method provides some further supporting evidence for the conjecture of Chang and Pettie. In particular, if one finds a $\poly \log n$-round 
deterministic algorithm for $\big(O(\log n), O(\log n)\big)$ network decomposition\cite{panconesi1992improved} --- a central problem that has remained open for a quarter century, but is often perceived as likely to be true --- then, combining that with our method would prove $T_{LLL}(n)=\poly(\log \log n)$. 
\medskip

\paragraph{A Gap in the Randomized Distributed Complexity Hierarchy}
Putting \Cref{RANDLLL} with \cite[Theorem 6]{chang2017time}, we get the following automatic speedup result:

\begin{restatable}{corollary}{Speedupp}
Let $\mathcal{A}$ be a randomized $\mathsf{LOCAL}$ algorithm that solves some LCL problem $\mathcal{P}$ on bounded degree graphs, w.h.p., in $o(\log n)$ rounds. Then, it is possible to transform $\mathcal{A}$ into a new randomized $\mathsf{LOCAL}$ algorithm $\mathcal{A}'$ that solves $\mathcal{P}$, w.h.p., in $2^{O(\sqrt{\log\log n})}$ rounds. 
 \end{restatable}

Using a similar method, and our deterministic LLL algorithm (\Cref{DETLLL}), we obtain the following corollary, the proof of which appears in \Cref{app:GeneralLLLAlgo}. This corollary shows that any $o(\log\log n)$-round randomized algorithm for an LCL problem on bounded degree graphs can be improved to a deterministic $O(\log^* n)$-round algorithm. This result seems to be implicit in the recent work of Chang, Kopelowitz, and Pettie \cite{chang2016exponential}, though with a quite different proof, and it can be derived from \cite[Corollary 3]{chang2016exponential} and \cite[Theorem 3]{chang2016exponential}.
\begin{corollary}
\label[corollary]{crl:loglogTologStar}
Let $\mathcal{A}$ be a randomized $\mathsf{LOCAL}$ algorithm that solves some LCL problem $\mathcal{P}$ on bounded degree graphs, w.h.p., in $o(\log\log n)$ rounds. Then, it is possible to transform $\mathcal{A}$ into a new deterministic $\mathsf{LOCAL}$ algorithm $\mathcal{A}'$ that solves $\mathcal{P}$ in $O(\log^* n)$ rounds. 
\end{corollary}

\smallskip
\paragraph{Faster Distributed Algorithms for Graph Colorings via LLL}
For some distributed graph problems on bounded degree graphs, we can immediately get faster algorithms by applying \Cref{RANDLLL}. 
However, there are two quantifiers which appear to limit the applicability of \Cref{RANDLLL}: (L1) it requires a stronger form of the LLL criterion, concretely needing $p(ed)^{32}<1$ instead of $epd\leq 1$; (L2) it applies mainly to bounded degree graphs.

We explain how to overcome these two limitations in most of the LLL-based problems studied by Chung, Pettie, and Su\cite{chung2014LLL}. Regarding limitation (L1), we show that even though in many coloring problems the direct LLL formulation 
would not satisfy the polynomial criterion $p(ed)^{32}<1$, we can still solve the problem, through a number of iterations of partial colorings, each of which satisfies this stronger LLL criterion. Regarding limitation (L2), we explain that in most of these coloring problems, the first step of our LLL algorithm, which is its only part that relies on bounded degrees, can be replaced by a faster randomized step, suited for that coloring. 

The end results of our method include algorithms with round complexity $2^{O(\sqrt{\log \log n})}$ for a number of coloring problems, improving on the corresponding $O(\log n)$-round algorithms of Chung, Pettie, and Su\cite{chung2014LLL}: defective coloring, frugal coloring, and list vertex-coloring, which are presented, respectively, in \Cref{sec:DefectiveColoring}, \Cref{sec:FrugalColoring}, and \Cref{app:listColoring}.
%




\vspace{-5pt}
\section{Preliminaries}\label{sec:prelim}
\vspace{-5pt}
\subsection{Network Decompositions}\label{subsec:NetDecomp}
Roughly speaking, a \emph{network decomposition}\cite{awerbuch1989network, panconesi1992improved} partitions the nodes into a few blocks, each of which is made of a number of low-diameter connected components. More formally, the definition is as follows:

\begin{definition}[Network Decomposition]
Given a graph $G=(V, E)$, a partition of the nodes $V$ into $C$ vertex-disjoint blocks $V_1$, $V_2$, \dots, $V_C$ is a $(C,D)$ network decomposition if in each block's induced subgraph $G[V_i]$ each connected component has diameter at most $D$. 
\end{definition}

\begin{lemma}[The Network Decomposition Algorithm]\label[lemma]{NWD} Given an $n$-node network $G=(V, E)$, there is a deterministic distributed algorithm that computes a $(\lambda, n^{1/\lambda}\cdot\log n)$ network decomposition of $G$ in $\lambda \cdot n^{1/\lambda}\cdot 2^{O(\sqrt{\log n})}$ rounds.
\end{lemma}
The proof of \Cref{NWD} is deferred to \Cref{App:NetDecomp}; it works mainly by putting together some ideas of Awerbuch and Peleg~\cite{awerbuch1990sparse}, Panconesi and Srinivasan\cite{panconesi1992improved}, and Awerbuch et al.\cite{awerbuch1996fast}. However, we are not aware of this result appearing in prior work.

\subsection{Shattering}
%
In a number of our algorithms, we make use of the following lemma which, roughly speaking, shows that if each node of the graph remains with some small probability and we have certain independence between these events, the remaining connected components are ``small''. We remark that this lemma or its variants are key ingredients in Beck's LLL method\cite{beck1991LLL}, sometimes referred to as the \emph{shattering lemma}, and analogues of it appear in \cite{alon1991LLL, michael2002graph, molloy1998LLL, rubinfeld2011fast, alon2012LCA, barenboim2012locality, Ghaffari-MIS, harris2016distributed, GS17}.
The proof of this lemma is postponed to \Cref{app:shattering}.

\begin{lemma}[The Shattering Lemma]\label[lemma]{Shattering}
Let $G=(V, E)$ be a graph with maximum degree $\Delta$. Consider a process which generates a random subset $B \subseteq V$ where $P(v \in B)\leq \Delta^{-c_1}$, for some constant $c_1 \geq 1$, and that the random variables $1(v\in B)$ depend only on the randomness of nodes within at most $c_2$ hops from $v$, for all $v\in V$, for some constant $c_2\geq 1$.  
Moreover, let $H=G^{[2c_2+1,4c_2+2]}$ be the graph which contains an edge between $u$ and $v$ iff their distance in $G$ is between $2c_2+1$ and $4c_2+2$. Then with probability at least $1- n^{-c_3}$, for any constant $c_3 < c_1 - 4c_2-2$, we have the following three properties:
\begin{enumerate}[leftmargin=1.2cm]
\item[(P1)] $H[B]$ has no connected component $U$ with $|U|\geq \log_{\Delta}n$.
\item[(P2)]  $G[B]$ has size at most $O\paren{ \log_{\Delta} n \cdot \Delta^{2c_2}}$.
\item[(P3)] Each connected component of $G[B]$ admits a $(\lambda, O(\log^{1/\lambda} n\cdot \log^2\log n))$ network decomposition, for any integer $\lambda\geq 1$, which can be computed in $\lambda \cdot \log^{1/\lambda} n \cdot 2^{O(\sqrt{\log\log n})}$ rounds, deterministically.
\end{enumerate}
\end{lemma}

\section{Our General Algorithm for Lov\'{a}sz Local Lemma}\label[section]{sec:GeneralLLLAlgo}
In this section, we explain our sublogarithmic-time LLL algorithm of \Cref{RANDLLL}, which solves LLL in $2^{O(\sqrt{\log\log n})}$ rounds on bounded degree graphs, given the condition that $p(ed)^{32}<1$. This algorithm is developed in two stages, as we overview next. 

In the first stage, presented in \Cref{subsec:baseLLL}, we explain a randomized algorithm with complexity $\lambda\cdot \log^{1/\lambda} n \cdot 2^{O(\sqrt{\log\log n})}$, given that an LLL criterion $p(ed)^{4\lambda}<1$ is satisfied. In the main regime of interest, the best LLL criterion exponent that we will assume is $\lambda=O(1)$, and thus this $(\lambda \cdot \log^{1/\lambda} n \cdot 2^{O(\sqrt{\log\log n})})$-round algorithm, on its own, would not get us to our target complexity of $2^{O(\sqrt{\log \log n})}$, although still being an improvement on the $O(\log n)$-round algorithm of \cite{chung2014LLL}. 

In the second stage, presented in \Cref{sec:bootstrapping}, we improve this complexity to $2^{O(\sqrt{\log\log n})}$. That improvement works mainly by viewing the sublogarithmic-time local algorithm of \Cref{subsec:baseLLL} as setting up a new LLL, with a much larger exponent $\lambda$ in its LLL criterion, hence allowing us get to a much smaller complexity by (recursively) applying the same scheme. This speed up is inspired by the ideas of Chung and Pettie\cite{chang2017time} which showed that LLL can be used to speed up sublogarithmic-time local algorithms.\footnote{Though, we find this recursive application of the idea to speed up the complexity of LLL itself, through increasing the exponent of the corresponding LLL criterion, somewhat amusing.}


\vspace{-2pt}
\subsection{The Base LLL Algorithm}\label{subsec:baseLLL}
\vspace{-2pt}
\begin{theorem}\label[theorem]{thm:baseLLL}
For any integer $\lambda\geq 8$, there is a randomized distributed algorithm solving the LLL problem under the symmetric criterion $p(ed)^{4\lambda}<1$, in $O(d^2) + \lambda \cdot \log^{1/\lambda} n\cdot 2^{O(\sqrt{\log \log n})}$ rounds, w.h.p. 
\end{theorem}

This algorithm consists of two parts: (1) a randomized algorithm, explained in \Cref{subsubsec:randomSamplingLLL}, which performs some partial sampling in the LLL space, thus setting some of the variables, in a manner that \emph{shatters} the graph, hence leaving small connected components among the unset variables; (2) a deterministic LLL algorithm, explained in \Cref{subsubsec:DETSamplingLLL}, which we use to solve the remaining small connected components. To the best of our knowledge, this is the first non-trival deterministic distributed LLL algorithm. In \Cref{wrapupBase}, we combine these two parts, concluding the proof of \Cref{thm:baseLLL}.

We note that this general \emph{shattering} style for randomized algorithms---which first performs some randomized steps to break the graph into small remaining connected components, and then uses some deterministic algorithm to solve these remaining components---is rooted in the breakthrough LLL algorithm of Beck\cite{beck1991LLL}, and has been used before also by \cite{alon1991LLL, michael2002graph, molloy1998LLL, barenboim2012locality, Ghaffari-MIS, harris2016distributed, GS17}. 


\subsubsection{The Randomized Part}\label{subsubsec:randomSamplingLLL}

We now explain the randomized component of our LLL algorithm for bounded degree graphs, which performs a partial sampling in the LLL space, thus setting some of the variables, in a manner that guarantees the following two properties: (1) the conditional probabilities of the bad events, conditioned on the already set variables, satisfy a polynomial LLL criterion, (2) the connected components of the events on variables that remain unset are ``small'' (e.g., for bounded degree graphs, they have size at most $O(\log n))$, with high probability.  

These two properties together will allow us to invoke the deterministic LLL algorithm that we present later in \Cref{subsubsec:DETSamplingLLL} on the remaining components of variables that remain unset. 
In particular, (1) means that the bad events $\mathcal{X}$ form another LLL problem on the variables that remain unset, where each new bad event has probability at most $\sqrt{p}$. Furthermore, (2) ensures that the components are small enough to make the deterministic algorithm efficient. 
 
Our partial sampling is inspired by a sequential LLL algorithm of Molloy and Reed\cite{molloy1998LLL}. 

\begin{lemma}[Random Partial Setting for the LLL Variables]\label[lemma]{RandomPartialSetting}
There is a randomized distributed algorithm that computes, w.h.p., in $O(d^2 + \log^* n)$ rounds, a partial assignment of values to variables --- setting the values of the variables in a set $\mathcal{V}^*\subseteq \mathcal{V}$, hence leaving the variables in $\mathcal{V}':=\mathcal{V} \setminus \mathcal{V}^*$ unset --- of an LLL satisfying $p(ed)^{4\lambda} < 1$, for any integer $\lambda\geq 8$, such that
\begin{enumerate}[(i)]
\item $Pr[A\mid \mathcal{V}^*]\leq \sqrt{p}$ for all $A \in \mathcal{X}$, and
\item w.h.p. each connected component of $G^2_{\mathcal{X}}[\mathcal{V'}]$ admits a $(\lambda,O(\log^{1/\lambda}n \cdot \log^2 \log n))$ network decomposition, which can be computed in $\lambda \cdot \log^{1/\lambda} n \cdot 2^{O(\sqrt{\log \log n})}$ rounds, deterministically.
\end{enumerate} 
\end{lemma} 
\begin{proof}[Proof]
We first compute a $(d^2+1)$-coloring of the square graph $G_{\mathcal{X}}^2$ on the events, which can be done even deterministically in $\tilde{O}(d) + O(\log^* n)$ rounds \cite{fraigniaud2016local}. Suppose $\mathcal{X}_i$ is the set of events colored with color $i$, for $i\in \{1, \dots, d^2+1\}$. We process the color classes one by one. 

For each color $i\in \{1, \dots, d^2+1\}$, and for each node $A\in \mathcal{X}_i$ in parallel, we make node $A$ sample values for its (non-frozen) variables locally, one by one. Notice that since we are using a coloring of $G^2_{\mathcal{X}}$, for each color $i$, each event $B\in \mathcal{X}$ shares variables with at most one event $A\in \mathcal{X}_i$. Hence, during this iteration, at most one node $A$ is sampling variables of event $B$. Each time, when node $A$ is choosing a value for a variable $v\in vbl(A)$, it checks whether this setting makes one of the events $B\in \mathcal{X}$ involving variable $v$ dangerous. We call an event $B$ \emph{dangerous} if $\Pr[B|\mathcal{V}^*_B ]\geq \sqrt{p}$, where $\mathcal{V}^*_B$ denotes the already set variables of $B$ up to this point in the sampling process. If the recently set variable $v$ leads to a dangerous event $B$, then we undo this variable assignment to $v$, and \emph{freeze} variable $v$ as well as all the remaining variables of event $B$. We will not assign any value to these frozen variables in the remainder of the randomized sampling process. We have two key observations regarding this process: 

\begin{observation}\label[observation]{obs1} At the end of each iteration, for each event $A \in \mathcal{X}$, the conditional probability of event $A$, conditioned on the already made assignments $\mathcal{V}^*_{A}$, is at most $\sqrt{p}$.
\end{observation}

\begin{observation}\label[observation]{obs2} For each event $A\in \mathcal{X}$, the probability of $A$ having at least one unset variable is at most $(d+1)\sqrt{p}$. Furthermore, this is independent of events that are further than $2$ hops from $A$. 
\end{observation}
\begin{proof}[Proof Sketch]
For each $B \in \mathcal{X}$, the probability that $B$ ever becomes dangerous is at most $\sqrt{p}$. This is because otherwise the total probability of $B$ happening would exceed $\sqrt{p}$. Now, an event $A\in \mathcal{X}$ can have frozen variables only if at least one of its neighboring events $B$, or event $A$ itself, becomes dangerous at some point during the process. Since $A$ has at most $d$ neighboring events, by a union bound, the latter has probability at most $(d+1)\sqrt{p}$.
\end{proof}

\Cref{obs1} directly implies property (i) of \Cref{RandomPartialSetting}. We use \Cref{obs2} to conclude that the events with at least one unset variable comprise ``small'' connected components. 
In particular, we apply \Cref{Shattering} to $G^2_{\mathcal{X}}$ with the random partial setting process generating a set $B\subseteq \mathcal{X}$ of the events that have at least one variable unset. By \Cref{obs2}, each event remains with probability at most $(d+1) \sqrt{p} \leq (d+1) \cdot e^{-2\lambda}\cdot  d^{-2\lambda}\leq d^{-15}$, hence we can set $c_1\gets 15$, and these events depend only on events within at most $1$ hop in $G_{\mathcal{X}}$, and hence $c_2\gets 2$ hops in $G^2_{\mathcal{X}}$. \Cref{Shattering} (P3) thus shows that with probability at least $1-n^{-5}$ property (ii) holds. 
%
\end{proof}

\subsubsection{The Deterministic Part}\label{subsubsec:DETSamplingLLL}

\begin{theorem}\label[theorem]{DETLLL}
For any integer $\lambda\geq 1$, any $n$-node distributed LLL problem can be solved deterministically in $\lambda \cdot n^{1/\lambda }\cdot 2^{O(\sqrt{\log n})}$ rounds, under the symmetric LLL criterion $p(ed)^\lambda < 1$. If the algorithm is provided a $(\lambda, \gamma)$ network decomposition of the square graph $G_{\mathcal{X}}^2$, then the LLL algorithm runs in just $O(\lambda \cdot (\gamma+1))$ rounds.
\end{theorem}
This algorithm makes use of \emph{network decompositions}. In particular, we will make a black-box invocation to the distributed algorithm stated in \Cref{NWD} for computing a $(\lambda, n^{1/\lambda}\cdot \log n)$ network decomposition, and then solve the LLL problem on top of this decomposition, by going through its blocks one by one. 

\paragraph{A Side Remark} The running time of our deterministic LLL algorithm hence directly depends on the network decomposition it works with. In particular, if there is a $\poly\log n$-round deterministic distributed algorithm that computes a $(\poly \log n, \poly \log n)$ network decomposition, then using the methods of \Cref{NWD}, we can obtain a $(\lambda, n^{1/\lambda}\log n)$-network decomposition in $n^{1/\lambda}\poly(\log n)$ rounds. Putting that together with the general methodology of our LLL algorithm\footnote{This connection is somewhat involved and has a number of parts. Thus we do not describe the details. One can verify it by tracing the steps of the development of \Cref{RANDLLL} and seeing how this improved network decomposition would improve the bounds, first in the base LLL algorithm, and then consequently in the bootstrapped LLL algorithm.} would prove that $T_{LLL}(n)=\poly(\log \log n)$, thus almost confirming the conjecture of Chang and Pettie \cite{chang2017time}. 

In fact, we believe that a conceivable future improvement of our LLL algorithm may need to improve this deterministic LLL algorithm, ideally to complexity $O(\log n)$, for proving the $T_{LLL}(n)=O(\log\log n)$ conjecture of Chang and Pettie \cite{chang2017time}. 
%
%
%
%


\begin{proof}[Proof of \Cref{DETLLL}]
We first compute a $(\lambda, n^{1/\lambda}\cdot \log n)$ network decomposition of $G_{\mathcal{X}}^2$, which decomposes its nodes into $\lambda$ disjoint blocks $\mathcal{X}_1, \dots, \mathcal{X}_{\lambda}$, such that each connected component of $G_{\mathcal{X}}^2[\mathcal{X}_i]$ has diameter at most $n^{1/\lambda} \cdot \log n$. This decomposition can be computed in $\lambda \cdot n^{1/\lambda}\cdot  2^{O(\sqrt{\log n})}$ rounds, using \Cref{NWD}. The rest of the proof is described assuming this $(\lambda, n^{1/\lambda}\cdot \log n)$ network decomposition and works in $O(\lambda \cdot n^{1/\lambda}\cdot \log n)$ rounds; one can easily see that given a $(\lambda, \gamma)$ network decomposition $G_{\mathcal{X}}^2$, the algorithm would work instead in $O(\lambda \cdot (\gamma+1))$ rounds.

Iteratively for $i=1, \dotsc, \lambda$, we assign values to all variables of events in $\mathcal{X}_i$ that have remained unset. The values are chosen is such a way that, after $i$ steps, the conditional probability of \emph{any} event in $\mathcal{X}$, conditioned on all the assignments in variables of events in $\bigcup_{j=1}^i \mathcal{X}_j$, is at most $p(ed)^i <1$. Once $i=\lambda$, since the conditional failure probability is $p(ed)^\lambda <1$ but all the variables are already assigned, we know that none of the events occurs. 

The base case $i=0$ is trivial. In the following, we explain how to set the values for variables involved in events of $\mathcal{X}_i$ in $n^{1/\lambda} \cdot \log n$ rounds. Let $\mathcal{V}_i$ be the set of variables in events of $\mathcal{X}_i$ that remain with no assigned value. We form a new LLL problem, as follows: For each bad event $A\in \mathcal{X}$, we introduce an event $B_{A,i}$ on the space of values of $\mathcal{V}_i$. This is the event that the values of $\mathcal{V}_i$ get chosen set such that the conditional probability of the event $A$, conditioned on the variables in $\bigcup_{j=1}^{i} \mathcal{V}_j$, is larger than $p (ed)^i$. Notice that $Pr[B_{A,i} \mid \bigcup_{j=1}^{i-1} \mathcal{V}_j] \leq \frac{p(ed)^{i-1}}{p (ed)^i} = \frac{1}{ed}$. Moreover, each event $B_{A, i}$ depends on at most $d$ other events $B_{A',i}$. Hence, the family of events $B_{A, i}$ on the variable set $\mathcal{V}_i$ satisfies the conditions of the tight (symmetric) LLL. Therefore, by the Lov\'{a}sz local lemma, we know that there exists an assignment to variables of $\mathcal{V}_i$ which makes no event $B_{A,i}$ happen. That is, an assignment such that the conditional probability of each event $A$, conditioned on the assignments in $\bigcup_{j=1}^i \mathcal{V}_j$, is bounded by at most $p (ed)^i$. 

Given the existence, we find such an assignment in $n^{1/\lambda}\cdot \log n$ rounds, as follows: each component of $G_{\mathcal{X}}^2[\mathcal{X}_i]$ first gathers the whole topology of this component (as well as its incident events and the current assignments to any of their variables), in $n^{1/\lambda}\log n$ rounds. Then, it decides about an assignment for its own variables in $\mathcal{V}_i$, by locally brute-forcing all possibilities. Different components can decide independently as there is no event that shares variables with two of them, since they are non-adjacent in $G_{\mathcal{X}}^2$.
%
%
\end{proof}

\subsubsection{Wrap-Up: Base LLL Algorithm}\label{wrapupBase}
\begin{proof}[Proof of \Cref{thm:baseLLL}]
We run the randomized algorithm of \Cref{RandomPartialSetting} for computing a partial setting of the variables, in $O(d^2+ \log^* n)$ rounds. Then, by \Cref{RandomPartialSetting} (i), the remaining events $\mathcal{X}'$ (those which have at least one unset variable) form a new LLL system on the unset variables, where each bad event has probability at most $\sqrt{p}$. 

Moreover, by \Cref{RandomPartialSetting} (ii), 
each connected component of the square graph $G^2_{\mathcal{X}}[\mathcal{X}']$ of these remaining events $\mathcal{X}'$ has a $(\lambda, O(\log^{1/\lambda} n \cdot \log^2\log n))$ network decomposition, which we can compute in $\lambda \cdot  \log^{1/\lambda} n \cdot 2^{O(\sqrt{\log \log n})}$ rounds, deterministically. From now on, we handle the remaining events in different connected components of $G^2_{\mathcal{X}}[\mathcal{X}']$  independently. 

Since $\sqrt{p} (ed)^{\lambda}<1$, we can now invoke the deterministic LLL algorithm of \Cref{DETLLL} on top of the network decomposition of each component. Our deterministic LLL then runs in $\lambda\cdot \log^{1/\lambda} n \cdot \log^2 \log n$ additional rounds, and finds assignments for these remaining variables, without any of the events occurring, hence solving the overall LLL problem. The overall round complexity is $O(d^2) + \lambda \cdot \log^{1/\lambda} n \cdot 2^{O(\sqrt{\log \log n})}$. 
\end{proof}

\vspace{-5pt}
\subsection{Improving the Base LLL Algorithm via Bootstrapping}\label{sec:bootstrapping}
\label{subsec:bootstrappedLLL}
\vspace{-3pt}

\begin{proof}[Proof of \Cref{RANDLLL}]
In \Cref{thm:baseLLL}, we saw an algorithm $\mathcal{A}$ that solves any $n$-event LLL under the criterion $p(ed)^{32} < 1$ in $T_{n, d} = O(d^2+ \log^{1/4} n)$ rounds. We now explain how to bootstrap this algorithm to run in $2^{O(\sqrt{\log\log n})}$ rounds, on bounded degree graphs.

Inspired by the idea of Chang and Pettie~\cite{chang2017time}, we will lie to $\mathcal{A}$ and say that the LLL graph has $n^* \ll n$ nodes, for a value of $n^*$ to be fixed later. Then, $\mathcal{A}_{n^*}$ runs in  $T_{n^*, d} = O(d^2+ \log^{1/4} n^*)$ rounds. In this algorithm, the probability of any local failure (i.e., a bad event of LLL happening) is at most $1/n^*$. We can view this as a new system of bad events which satisfies a much stronger LLL criterion. In particular, we consider each of the previous bad LLL events as a bad event of the new LLL system, on the space of the random values used by $\mathcal{A}_{n^*}$, but now we connect two bad events if their distance is at most $2T_{n^*, d}+1$. Notice that if two events are not connected in this new LLL, then in algorithm $A_{n^*, d}$, they depend on disjoint sets of random variables and thus they are independent. 

The degree of the new LLL system is $d' = d^{2T_{n^*, d}+1} = d^{O(d^2+ \log^{1/4} n^*)}$. On the other hand, the probability of the bad events of the new system is at most $p'=1/n^*$. Hence, the polynomial LLL criterion is satisfied with exponent $\lambda' = \frac{\log_{d} n^*}{O(d^2 + \log^{1/4} n^*)}$. We choose $n^* = \log n$, which, for $d=O((\log\log n)^{1/5})$, means $\lambda' = \Omega(\sqrt{\log\log n})$. Hence, this new LLL system can be solved using the LLL algorithm of \Cref{thm:baseLLL} in time 
\begin{align*}
& (d')^2 + \lambda'\cdot  \log^{1/\lambda'} n \cdot 2^{O(\sqrt{\log\log n})} = \\ 
& d^{O(d^2 + (\log\log n)^{1/4})} + \sqrt{\log\log n} \cdot (\log n)^{1/\Omega(\sqrt{\log\log n})} \cdot 2^{O(\sqrt{\log\log n})} = 2^{O(\sqrt{\log\log n})}.
\end{align*}
We should note that these are rounds on the new LLL system, but each of them can be performed in $2T_{n^*, d}+1 = O(d^2+ \log^{1/4} n^*) = O(\sqrt{\log\log n})$ rounds on the original graph. Hence, the overall complexity is still $2^{O(\sqrt{\log\log n})}$.
\end{proof}

We next state another result obtained via this speedup method, targeting higher degree graphs, which we will use in our coloring algorithms. The proof is deferred to \Cref{app:GeneralLLLAlgo}.  
\begin{lemma}\label[lemma]{lemma:SpecialSpeedup}
Let $\mathcal{A}$ be a randomized $\mathsf{LOCAL}$ algorithm that solves some LCL problem $\mathcal{P}$ on $n$-node graphs with maximum degree $d \leq 2^{O(\log^{1/4}\log n)}$ in $O(\log^{1/4} n)$ rounds. Then, it is possible to transform $\mathcal{A}$ into a new randomized $\mathsf{LOCAL}$ algorithm $\mathcal{A}'$ that solves $\mathcal{P}$, w.h.p., in $2^{O(\sqrt{\log\log n})}$ rounds.
\end{lemma}

%
%




\section{Defective Coloring}
\label{sec:DefectiveColoring}

An $f$-defective coloring is a (not necessarily proper) coloring of nodes, where each node has at most $f$ neighbors with the same color. In other words, in an $f$-defective coloring, each color class induces a subgraph with maximum degree $f$. Chung, Pettie, and Su \cite{chung2014LLL} gave an $O(\log n)$-round distributed algorithm for computing an $f$-defective coloring with $O(\Delta/f)$ colors. We here improve this complexity to $2^{O\paren{\sqrt{\log \log n}}}$ rounds.

\begin{theorem}\label[theorem]{defective-col}
There is a $2^{O\paren{\sqrt{\log \log n}}}$-round randomized distributed algorithm that computes an $f$-defective $O(\Delta/f)$-coloring in an $n$-node graph with maximum degree $\Delta$, w.h.p., for any integer $f \geq 0$.
\end{theorem}

\paragraph{Direct LLL Formulation of Defective Coloring}
Chung, Pettie, and Su \cite{chung2014LLL} give a formulation of $f$-defective $\ceil{2\Delta/f}$-coloring as LLL as follows. Each node picks a color uniformly at random. For each node $v$, there is a bad event $D_v$ that $v$ has more than $f$ neighbors assigned the same color as $v$. The probability of a neighbor $u$ having the same color as $v$ is $f/(2\Delta)$. Hence, the expected number of neighbors of $v$ with the same color as $v$ is at most $f/2$. By a Chernoff bound, the probability of $v$ having more than $f$ neighbors with the same color is at most $e^{-f/6}$. Moreover, the dependency degree between the bad events $D_v$ is $d \leq \Delta^2$. Therefore, $p(ed)^{32} \leq e^{-f/6+32+64 \log \Delta}< 1$ for $f=\Omega( \log \Delta)$. 

We are unable to directly apply our LLL algorithm of \Cref{RANDLLL} to this formulation, because: (A) For $f = o(\log \Delta)$, this LLL formulation does not satisfy the polynomial criterion $p(ed)^{32}<1$, (B) even if this criterion is satisfied, the dependency degree $d$ may be larger than what \Cref{RANDLLL} can handle. 


%

\paragraph{Iterative LLL Formulation of Defective Coloring  via Bucketing}
Instead of directly finding an $f$-defective $O(\Delta/f)$-coloring with one LLL problem---i.e., a partition of $G$ into $O(\Delta/f)$ buckets with maximum degree $f$ each ---we gradually approach this goal by iteratively partitioning the graph into buckets, until they have maximum degree $f$. In other words, we slow down the process of partitioning. We gradually decrease the degree, moving from maximum degree $x$ to $\log^5 x$ in one iteration. We can see each of these bucketing steps --- that is, the partitioning into subgraphs --- as a partial coloring, which fixes some bits of the final color. 
Each of these slower partitioning steps can be formulated as an LLL. The function $x \mapsto \log^5x$ is chosen large enough for the corresponding LLL to satisfy the polynomial criterion, and small enough so that decreasing the degree from $\Delta$ to $f$ does not take too many iterations, namely $O(\log^*\Delta)$ iterations only.



We now explain how a defective coloring problem can be solved using iterated bucketing. We first formulate the bucketing as an LLL problem satisfying the polynomial LLL criterion, and present ways for solving this LLL for different ranges of $\Delta$. Then, we explain how iterated application of solving these bucketing LLLs leads to a partition of the graph into $O(\Delta/f)$ many degree-$f$ buckets. 

\paragraph{One Iteration of Bucketing}
In one bucketing step, we would like to partition our graph with degree $\Delta$ into roughly $\Delta/\Delta'$ buckets, each with maximum degree $\Delta'$, for a $\Delta'=\Omega(\log^5 \Delta)$. Notice that we can achieve the defective coloring of \Cref{defective-col}, by repeating this bucketing procedure, iteratively. See the proof of \Cref{defective-col}, which appears in \Cref{app:DefectiveColoring}, for details of iterative bucketing. Each iteration of bucketing can be formulated as an LLL as follows.

%
%
 %

\paragraph{LLL Formulation of Bucketing} Let $k=(1+\eps)\Delta/\Delta'$ for $\eps=\log^{2} \Delta/\sqrt{\Delta'}$. We consider the random variables assigning each node a bucket number in $[k]$. Then, we introduce a bad event $D_v$ for node $v$ if more than $\Delta'$ neighbors of $v$ are assigned the same number as $v$. In expectation, the number of neighbors of a node in the same bucket is at most $\Delta'/(1+\eps)$. By a Chernoff bound, the probability of having more than $\Delta'$ neighbors in the same bucket is at most $p= e^{-\Omega(\eps^2 \Delta')}=e^{-\Omega(\log^4 \Delta)}$. Moreover, the dependency degree between these bad events is $d \leq \Delta^2$. Hence, this LLL satisfies the polynomial criterion.

If $\Delta \leq O(\log^{1/10} \log n)$, then $d=O(\log^{1/5}\log n)$, and thus we can directly apply the LLL algorithm of \Cref{RANDLLL} to compute such a bucketing in $2^{O(\sqrt{\log\log n})}$ rounds. For larger values of $\Delta$, however, we cannot apply \Cref{RANDLLL} directly. The following lemma, the proof of which is deferred to \Cref{app:DefectiveColoring}, discusses how we handle this range by sacrificing a $2$-factor in the number of buckets. In a nutshell, the idea is to just perform one sampling step of bucketing, and then to deal with nodes with too large degree separately, by setting up another bucketing LLL. While the first LLL on the whole graph could not be solved directly, the second LLL is formulated only for a ``small'' subset of nodes, which allows an efficient solution. Because of the two trials of solving an LLL, we lose a $2$-factor in the total number of buckets.

\begin{lemma}\label[lemma]{Bucketing}
For $\Delta \geq \Omega({\log^{1/10}\log n})$, there is a $2^{O(\sqrt{\log \log n})}$-round randomized distributed algorithm that computes a bucketing into $2k$ buckets with maximum degree $\Delta'$ each, for $\Delta'=\Omega(\log^5 \Delta)$,  $\eps=\log \Delta/\sqrt{\Delta'}$, and $k=(1+\eps)\Delta/\Delta'$, with high probability.
\end{lemma}
\section{Frugal coloring}
\label{sec:FrugalColoring}

A \emph{$\beta$-frugal coloring} is a proper coloring in which no color appears more than $\beta$ times in the neighborhood of any node. 
%
%
We improve the complexity of $\beta$-frugal $O(\Delta^{1+1/\beta})$-coloring from $O(\log n)$ by Chung, Pettie, and Su \cite{chung2014LLL} to $2^{O(\sqrt{\log \log n})}$.

\begin{theorem}\label[theorem]{frugal}
There is a $2^{O(\sqrt{\log \log n})}$-round randomized distributed algorithm for a $\beta$-frugal $(120 \cdot \Delta^{1 + 1/\beta})$-coloring\footnote{We remark that we have not tried to optimize this constant $120$.} in a graph with $n$ nodes and maximum degree $\Delta$, w.h.p., for any integer $\beta \geq 1$. \end{theorem}



\paragraph{Direct LLL Formulation of Frugal Coloring}
Molloy and Reed \cite[Theorem 19.3]{michael2002graph} formulated frugal coloring as an LLL problem in the following way: Each node picks a color uniformly at random. There are two types of bad events: On the one hand, we have the \emph{properness} condition, i.e., a bad event $M_{u,v}$, for each $\{u,v\}\in E$, which happens if $u$ and $v$ have the same color. On the other hand, the \emph{frugality} condition --- requiring that no node has more than $\beta$ neighbors of the same color. That is, we have one bad event $F_{u_1, \dotsc, u_{\beta +1}}$ for each set $u_1, \dotsc, u_{\beta + 1} \in N(v)$ of nodes in the neighborhood of some node $v$, which happens if all these nodes $u_1, \dotsc, u_{\beta + 1}$ are assigned the same color. For palettes of size $C$, the probability of a bad event is at most $1/C$ and $1/C^{\beta}$ for type 1 and type 2, respectively. Each event depends on at most $(\beta+1)\Delta$ type 1 and at most $(\beta +1) \Delta {{\Delta} \choose \beta}$ type 2 events.

\paragraph{Iterated LLL Formulation of Frugal Coloring via Partial Frugal Coloring}
While the above formulation is enough to satisfy the asymmetric tight LLL criterion for $C=O(\Delta^{1+1/\beta})$, it does not satisfy the (symmetric) polynomial LLL.
Therefore, the algorithm of \Cref{RANDLLL} is not directly applicable. 
We show how to break down the frugal coloring problem into a sequence of few partial coloring problems, coloring only some of the nodes that have remained uncolored, each of them satisfying the polynomial LLL criterion. 



\paragraph{Roadmap} In \Cref{sec:sampling-frugal}, we formalize our notion of partial frugal colorings and present a method for sampling them. Then, in \Cref{LLLPartialFrugal}, we show how to use this sampling to formulate the problem of finding a partial frugal coloring guaranteeing progress (to be made precise) as a polynomial LLL and how to solve it. In \Cref{sec:frugal-completing}, we explain how --- after several iterations of setting up and solving these ``progress-guaranteeing'' LLLs, gradually extending the partial frugal coloring --- we can set up and solve one final polynomial LLL for completing the partial coloring, also based on the sampling method presented in \Cref{sec:sampling-frugal}. 

\subsection{Sampling a Partial Frugal Coloring}\label{sec:sampling-frugal}
\begin{definition}[Partial Frugal Coloring]
A partial $\beta$-frugal coloring of $G=(V,E)$ is an assignment of colors to a subset $V^* \subseteq V$ such that it is proper in $G[V^*]$ and no node in $V$ has more than $\beta$ neighbors with the same color. In other words, it is a $\beta$-frugal coloring of $G[V^*]$ with the additional condition that no uncolored node in $V' :=V\setminus V^*$ has more than $\beta$ neighbors in $V^*$ with the same color. 
\end{definition}
A partial coloring naturally splits the base graph $G$ into two parts: $G[V^*]$ induced by colored nodes and $G[V']$ induced by uncolored nodes. However, the problem of extending or completing a partial frugal coloring does not only depend on $G[V']$, but also on the base graph $G$. That is why we introduce the notion of base-graph degree, a property of the uncolored set $V'$ with respect to the base graph $G$. 

%
%

\begin{definition}[Base-Graph Degree of a Partial (Frugal) Coloring]
Given a partial coloring, we call the number $d(v,V')$ of neighboring uncolored nodes of a node $v \in V$ its \emph{base-graph degree} into the uncolored set $V'$. Moreover, we call the maximum base-graph degree $\Delta'$ of a node $v \in V$ into $V'$ the base-graph degree of $V'$. 
\end{definition}
%
In the following, we show how one can sample a partial frugal coloring, thus randomly assign some of the nodes in a set $V'$ of uncolored nodes a color. The main idea of our sampling process is to pick a color uniformly at random, and then discard it if this choice would lead to a violation (in terms of properness and frugality). In order to increase the chances of a node being colored, instead of just sampling one color, each node $v$ samples $x$ different colors from $x$ different palettes at the same time, for some parameter $x\geq 1$, and then picks the first color that does not lead to a violation. If $v$ has no such violation-free among its $x$ choices, then $v$ remains uncolored. 

The next lemma, the proof of which is deferred to \Cref{app:FrugalColoring}, analyzes the probability of two kinds of events: Event (E1) that a node is uncolored. This event is important if we aim to color all the nodes in $V'$. Event (E2) that the base-graph degree of a node into the set of uncolored nodes in $V'$ is too large. This event is important if we do not aim at a full coloring of all the nodes in $V'$, but we want to ensure that we make enough progress in decreasing the base-graph degree of the uncolored set.  
\begin{lemma}\label[lemma]{frugal-iteration}
Let $G=(V,E)$ be a graph with maximum degree $\Delta$, $V' \subseteq V$ an uncolored set with base-graph degree $\Delta'$,  
$\beta \in [\Delta]$, and $x \geq 1$. Then there is an $O(1)$-round randomized distributed algorithm that computes a partial $\beta$-frugal $(20 \cdot x  \cdot \Delta' \cdot \Delta^{1/\beta})$-coloring of some of the nodes in $V'$ such that 
\begin{enumerate}[(i)]
\item the probability that a node in $V'$ is uncolored is at most $10^{-x}$,  
\item the probability that the base-graph degree of a node $v \in V$ into the uncolored subset of $V'$ is larger than $5^{-x}\cdot\Delta'$ is at most $e^{-\Omega(5^{-x} \cdot \Delta')}$. 
\end{enumerate} 
\end{lemma}


%
%
\subsection{Iterated Partial Frugal Coloring}\label{LLLPartialFrugal}

In the following, we first show how a ``progress-guaranteeing'' partial coloring --- that is, a coloring that decreases the base-graph degree of every node quickly enough --- can be found based on the sampling process presented in \Cref{sec:sampling-frugal}. 
Then, we prove that by iterating this algorithm for $O(\log^*\Delta)$ repetitions, using different palettes in each iteration, the base-graph degree reduces to $O(\sqrt{\Delta})$. 

In one iteration, given a set $V'$ of uncolored nodes, we want to color a subset $V^* \subseteq V'$ such that the uncolored nodes $V'':=V'\setminus V^*$ have a base-graph degree $\Delta''$ that is sufficiently smaller than the base-graph degree $\Delta'$ of $V'$. Note that the sampling of \Cref{sec:sampling-frugal} only provides us with a partial coloring where every node is likely to have a decrease in the base-graph degree. Here, however, we want to enforce that for every node in $V$ there is such a decrease. To this end, we set up an LLL as follows.

\paragraph{LLL Formulation for ``Progress-Guaranteeing'' Coloring}
 Performing the sampling of \Cref{frugal-iteration}, we have a bad event $D_v$ for every node $v \in V$ that its base-graph degree into $V''$ is larger than $\Delta''=5^{-x}\cdot \Delta'$. By \Cref{frugal-iteration} (ii), we know that the probability of $D_v$ is at most $e^{-\Omega(5^{-x} \cdot \Delta')}$. Moreover, the dependency degree is at most $d\leq \Delta^2$. This LLL thus satisfies the polynomial criterion.

However, as $d$ might be large, we cannot directly apply the LLL algorithm of \Cref{RANDLLL}. In the following, we present an alternative way of finding a partial coloring ensuring a drop in the base-graph degree of every node. In a nutshell, the idea is to just perform one sampling step of a partial frugal coloring, as described in \Cref{sec:sampling-frugal}, and then deal with nodes associated with bad events (to be made precise) separately, by setting up another ``progress-guaranteeing'' LLL. While the first LLL on the whole graph could not be solved directly, the second LLL is formulated only for a ``small'' subset of nodes, which allows an efficient solution. Because of the two trials of solving an LLL, we lose a $2$-factor in the total number of colors.  The proof of the next lemma appears in \Cref{app:FrugalColoring}.

\begin{lemma}\label[lemma]{lemma:frugal-iteration-sampling}
Given a partial $\beta$-frugal coloring with uncolored set $V'$ with base-graph degree $\Delta'$ and a parameter $x\geq 1$ such that $5^{-x} \cdot\Delta' = \Omega(\sqrt{ \Delta})$, there is a $2^{O(\sqrt{\log \log n})}$-round randomized distributed algorithm that computes a partial $\beta$-frugal $(40\cdot x \cdot \Delta' \cdot \Delta^{1/\beta})$-coloring such that the uncolored set has base-graph degree at most $\Delta'' = 5^{-x} \cdot \Delta'$. 
\end{lemma}




The next lemma describes how through iterated application of finding partial colorings, as supplied by \Cref{lemma:frugal-iteration-sampling}, the base-graph degree of the uncolored set decreases to $O(\sqrt{\Delta})$ after $O(\log^*\Delta)$ rounds and using at most $O(\Delta^{1 + 1/\beta})$ colors. The related proof appears in \Cref{app:FrugalColoring}.

\begin{lemma}\label[lemma]{lemma:whole-iterations}
There is a $2^{O(\sqrt{\log \log n})}$-round randomized algorithm that computes a partial $\beta$-frugal $(80\cdot\Delta^{1+1/\beta})$-coloring such that
the uncolored set $V'$ has base-graph degree $O(\sqrt{\Delta})$. 
\end{lemma}

\subsection{Completing a Partial Frugal Coloring}\label{sec:frugal-completing}
In this section, we describe how, once the base-graph degree is $O(\sqrt{\Delta})$, all the remaining uncolored nodes can be colored, hence completing the partial frugal coloring. We first give a general formulation for the completion of partial frugal colorings. 

\paragraph{LLL Formulation for Completion of Partial Frugal Coloring}
Performing the sampling of \Cref{frugal-iteration}, we have a bad event $U_v$ for every node $v \in V$ that it is uncolored. By \Cref{frugal-iteration} (i), the probability of $U_v$ is at most $10^{-x}$. Moreover, the dependency degree $d$ is at most $\Delta^2$. This LLL satisfies the polynomial criterion if $x=\Omega(\log \Delta)$. 

In the following lemma, the proof of which appears in \Cref{app:FrugalColoring}, we show to solve this LLL. The idea is to first perform one sampling step (of \Cref{frugal-iteration}), which shatters the graph into ``small'' components of uncolored nodes, then to set up an LLL for completing the partial coloring, and finally to solve it by employing our deterministic LLL algorithm, on each of the components. 

\begin{lemma}\label[lemma]{lemma:frugal-color-good}
Given a partial $\beta$-frugal coloring and a set $V'$ of uncolored nodes with base degree $\Delta'=O(\sqrt{\Delta})$, there is a $2^{O(\sqrt{\log \log n})}$-round randomized algorithm that completes this $\beta$-frugal coloring, by assigning colors to all nodes in $V'$, using $40\cdot \Delta^{1+ 1/\beta}$ additional colors.   
\end{lemma}

A wrap-up of these results about iterated partial colorings and completing a partial coloring immediately leads to a proof of \Cref{frugal}. 


\begin{proof}[Proof of \Cref{frugal}]
We first apply the iterated coloring algorithm of \Cref{lemma:whole-iterations} with $80\cdot \Delta^{1+1/\beta}$ colors, in $2^{O(\sqrt{\log \log n})}$ rounds. Then, we run the algorithms of \Cref{lemma:frugal-color-good} to complete this partial coloring with $40 \cdot \Delta^{1+1/\beta}$ additional colors, in $2^{O(\sqrt{\log \log n})}$ rounds. 
This yields a $\beta$-frugal $(120 \cdot \Delta^{1+1/\beta})$-coloring, in $2^{O(\sqrt{\log \log n})}$ rounds.   
\end{proof}


\section{List Vertex-Coloring}
\label{app:listColoring}
Given an $n$-node graph $G=(V, E)$, a list vertex-coloring assigns each node $v$ a color from its \emph{color list} $L_v$ such that no two neighboring nodes choose the same color. 
The color lists satisfy the following properties: (1) $|L_v|\geq |L|$ for all $v\in V$. We emphasize that the list size $L$ may be much smaller than the degree $\Delta$. (2) For each $v\in V$ and each color $q\in L_v$, the set $N_{q}(v) = \{u \,|\, u\in N(v) \,\&\, q\in L_u\}$ of neighbors $u$ of $v$ that also have color $q$ in their color list $L_u$ has size $|N_q(v)| \leq L/C$, for a given large constant $C > 2e$. 

Chung, Pettie, and Su \cite{chung2014LLL} gave an $O(\log n)$-round randomized distributed algorithm for list vertex-coloring with $C\geq 2e+\eps$. We here improve this complexity to $2^{O(\sqrt{\log \log n})}$ rounds, for a sufficiently large constant $C$. We remark that we have not tried optimizing this constant $C$. 

\begin{theorem}
There is a $2^{O(\sqrt{\log \log n})}$-round randomized distributed algorithm that w.h.p. computes a list vertex-coloring in an $n$-node graph where each list $L_v$ has size $L$ and for each color $q\in L_v$, we have we have $|N_q(v)|\leq L/C$, for some sufficiently large constant $C > 2e$.
\end{theorem}

\paragraph{Remark} Essentially without loss of generality, we can focus on the regime where $L=O(\log^2 n)$. This is because if $L=\Omega(\log^2 n)$, we can make each node choose $\log^2 n$ colors in its list at random to retain, forming a new color list $L'_v$ with $|L'_v|=\log^2 n$. Then, with high probability, we have the following property: for each node $v$ and each color $q\in L'_v$, the number of neighbors $u$ of $v$ that have $q\in L'_u$ is at most $(1+o(1))\log^2 n/C$.

\paragraph{Direct LLL Formulation of List Vertex-Coloring} Suppose each node picks a color uniformly at random. Define a bad event $E_{u, v, q}$ for each edge $\{u,v\}$ and color $q$ if its endpoints choose the same color $q$. The probability of each such event is at most $p=(\frac{1}{L})^2$. The dependency degree between these events is at most $d=2L\cdot L/C$. This is because, the event $E_{u, v, q}$ has dependency with the events of at most $L$ colors from each endpoint $u$ or $v$, and at most $L/C$ edges incident on that endpoint for each of these $L$ colors. Hence, if $C> 2e$, the LLL criterion $epd\leq 1$ is satisfied.  

\paragraph{Shortcomings of the Direct LLL Formulation} As before, we face two issues in applying the LLL algorithm \Cref{RANDLLL}: (1) the above formulation does not satisfy the polynomial LLL criterion, (2) the dependency degree $d$ may be above what \Cref{RANDLLL} can handle. 

\paragraph{Iterated LLL Formulation of List Vertex-Coloring via Pruning}
In the following, we explain how through a sequence of gradual pruning of color lists, we can get to our target coloring. A pruning step can be formulated as an LLL which satisfies the polynomial LLL criterion. We also explain how to perform each of these pruning steps, albeit the fact that the related (strengthened) LLL does not have bounded degrees. 

\subsection*{An LLL Formulation for Gradual Pruning of the Color Lists}

\paragraph{A $2$-Factor Pruning of Color Lists} We would like to narrow down the color lists and their conflict sizes by roughly a $2$ factor. More concretely, we would like that each node $v$ keeps a subset $L'_v\subseteq L_v$ such that $|L'_v|\geq \frac{|L_v|}{2} \cdot (1-\frac{1}{\log^2 L})$, and moreover, for each color $q \in L'_v$, the number of neighbors $u$ of $v$ that have $q\in L'_u$ is at most $(1+\frac{1}{\log^2 L})\frac{L}{2C}$. 

By repeating this $2$-factor pruning for (roughly) $\log_2 L/C=O(\log\log n)$ iterations, we get to a setting where each node has a color list of size at least $C/2$, none of which are kept by any neighbor. Then, each node can pick any of these colors as its final color. 

\paragraph{LLL for $2$-Factor Pruning of Color Lists} Let each node $v$ keep each of its colors in $L_v$ with probability $1/2$, forming its new list $L'_v$. We have two types of bad events, first that $|L'_v| \leq \frac{|L_v|}{2} \cdot (1-\frac{1}{\log^2 L})$, and second that for a color $q \in L'_v$, the number of neighbors $u$ of $v$ that have $q\in L'_u$ is more than $(1+\frac{1}{\log^2 L})\frac{L}{2C}$.
By Hoeffding bound, the probability of each of these bad events is at most $p=exp(-\Theta(\sqrt{L}/\log^2 L))$. Each event depends on at most $d=O(L^2/C)$ many others. Thus, this LLL satisfies the polynomial LLL criterion.

\medskip
For $L \leq O((\log\log n)^{1/10})$, we can directly apply the LLL algorithm of \Cref{RANDLLL} to solve the above $2$-factor pruning in $2^{O(\sqrt{\log\log n})}$ rounds. In \Cref{app:savingLLLpruning}, we explain how we solve the other cases of this pruning LLL in $2^{O(\sqrt{\log\log n})}$ rounds.

\subsection*{Acknowledgment} We are grateful to Seth Pettie for sharing an earlier manuscript of \cite{chang2017time} with us.
\bibliographystyle{plainurl}
\bibliography{../ref}

\appendix
\section{Missing Details of \Cref{sec:prelim}: Preliminaries}

\subsection{Deterministic Network Decomposition}\label{App:NetDecomp}
We here explain that by putting together some ideas and algorithm of Awerbuch and Peleg~\cite{awerbuch1990sparse}, Panconesi and Srinivasan\cite{panconesi1992improved}, and Awerbuch et al.\cite{awerbuch1996fast}, we can obtain the deterministic network decomposition algorithm claimed in \Cref{NWD}.

\begin{proof}[Proof of \Cref{NWD}]
We first describe a sequential algorithm for computing a $(\lambda, n^{1/\lambda}\cdot \log n)$ network decomposition, colloquially referred to as \emph{ball carving}. This technique was first presented by Awerbuch and Peleg~\cite{awerbuch1990sparse}. Then, we explain how to transform this sequential ball carving process into an efficient deterministic distributed algorithm, with round complexity $n^{1/\lambda}\cdot 2^{O(\sqrt{\log n})}$, using another network decomposition algorithm of Panconesi and Srinivasan\cite{panconesi1992improved}, and an idea of Awerbuch et al.\cite{awerbuch1996fast}.

\paragraph{Sequential Network Decomposition via Ball Carving} We will decompose the graph into vertex-disjoint blocks $A_1$ to $A_{\lambda}$, such that in the graph $G[A_i]$, each connected component has diameter at most $n^{1/\lambda}\cdot \log n$. We first describe the process of generating the first block $A_1$. The generation of the next blocks is similar.

We choose an arbitrary node $v$ to be the center of a new ball, and we use $B_r(v)$ denote the set of nodes with distance at most $r$ from $v$. Let $r^*$ be the smallest $r \geq 0$ for which $\frac{\left|B_{r}(v)\right|}{\left|B_{r-1}(v)\right|}< 1 + \frac{1}{n^{1/\lambda}}$. Since $\left| B_r(v)\right| \geq \left(1+ \frac{1}{n^{1/\lambda}}\right)^r$ for $r \leq r^*$, it must hold that $r^* \leq n^{1/\lambda}\log n$. We put all the nodes in $B_{r^*-1}(v)$ into $A_1$, and then delete $B_{r^*}(v)$ from $G$. That is, we carve and remove a ball of radius $r^*$ around the node $v$, but then only taking the inner part of it --- nodes, that are not on the boundary --- to $A_1$. Then, we pick another node in this remainder graph, and perform another ball carving; we repeat this process until no node is left. 

Once we are done with defining $A_1$, we remove all nodes of block $A_1$ from the graph $G$, and then work on the remaining graph $G_2=G\setminus A_1$. Then, we create the new block $A_2$, by a similar iterative ball-carving process on $G_2$. More generally, after computing blocks $A_1$ to $A_{i-1}$, we move to the graph $G_i=G\setminus \cup_{j=1}^{i-1} A_j$ and compute the block $A_i$, by a sequential ball-carving process similar to above. 

We now argue that $\lambda$ blocks exhaust the graph. In each iteration $i$ of computing another block, the size of the remaining graph shrinks such that $|G_i| < \frac{1}{n^{1/\lambda}} \cdot |G_{i-1}|$. This is because for each carved ball where we include $B_{r^*-1}(v)$ in $A_i$ and then discard the boundary $B_{r^*}(v)\setminus B_{r^*-1}(v)$, leaving them for the next blocks, we have $|B_{r^*}(v)\setminus B_{r^*-1}(v)| \leq \frac{1}{n^{\eps}} |B_{r^*-1}(v)|$. Since $|G_i| < \frac{1}{n^{1/\lambda}} \cdot |G_{i-1}|$, after $\lambda$ blocks, the remaining graph $G_{\lambda+1}$ is empty, and the process terminates.

\paragraph{Distributed Network Decomposition via Ball Carving} To compute the desired $(\lambda, n^{1/\lambda}\log n)$ network decomposition, we will simulate the ball carving idea explained above. However, we need to speed up the process, and make it run in $\lambda \cdot n^{1/\lambda}\cdot 2^{O(\sqrt{\log n})}$ rounds. For that, we use another network decomposition as a helper tool. In particular, we first compute a $\big(2^{O(\sqrt{\log n})},2^{O(\sqrt{\log n})}\big)$ network decomposition of $G^{d}$ for $d=2n^{1/\lambda}\cdot \log n+1$, by running the algorithm of Panconesi and Srinivasan \cite{panconesi1992improved} on $G^{d}$. This takes $d \cdot 2^{O(\sqrt{\log n})}$ 
rounds. It partitions the graph $G$ into $\ell=2^{O(\sqrt{\log n})}$ vertex-disjoint blocks $G_1, \dotsc, G_{\ell}$ such that for each block $G_i$, each connected component of $G_i$ has diameter at most $d \cdot  2^{O(\sqrt{\log n})}$ in the graph $G$, while any two components of $G_i$ are non-adjacent in $G^{d}$ and thus have distance at least $d+1$ in $G$. 

We now use this network decomposition to compute the desired output $(\lambda, n^{1/\lambda}\cdot \log n)$ network decomposition which partitions $V$ into vertex-disjoint sets $A_1$, $A_2$, \dots, $A_{\lambda}$ such that in each subgraph $G[A_i]$ for $i\in \{1, \dots, \lambda\}$, each connected component has diameter at most $n^{1/\lambda}\cdot \log n$. 
The construction is made of $\lambda$ epochs, each of which computes one of the blocks $A_i$, in $n^{1/\lambda}\cdot 2^{O(\sqrt{\log n})}$ rounds. We next discuss the first epoch, which computes the block $A_1$. The next epochs are similar, and compute the other blocks $A_2$ to $A_{\lambda}$, each repeating the procedure on the remaining graph.

\paragraph{Each epoch in the Construction of $\mathbf{A_1}$} The epoch is broken into $\ell=2^{O(\sqrt{\log n})}$ phases, each of which takes $n^{1/\lambda} \cdot2^{O(\sqrt{\log n})}$ rounds, hence making for an overall round complexity of $n^{1/\lambda}\cdot 2^{O(\sqrt{\log n})}$ for the epoch. We will simulate the ball-carving process of computing $A_1$, throughout these phases. During this process, each node is in one of the following three statuses: some nodes are put in $A_1$ (these are the inner parts of the carved balls), some nodes are processed and discarded (these are the boundaries of the carved balls), and some nodes are unprocessed. 

In the $j^{th}$ phase, we do as follows: Consider the set of nodes of $G_{j}$. Notice that each component of $G_{j}$ has diameter at most $d \cdot  2^{O(\sqrt{\log n})}$, and moreover, each two components have distance at least $d+1$. We first make the minimum-ID node of each of these components learn the $(n^{1/\lambda}\cdot \log n)$-neighborhood of its component, as well as the status of the nodes in this neighborhood. Notice that this information is within distance at most $d \cdot  2^{O(\sqrt{\log n})} + n^{1/\lambda}\cdot \log n =n^{1/\lambda}\cdot 2^{O(\sqrt{\log n})}$ 
from that minimum-ID node, and hence this can be done in $n^{1/\lambda} \cdot 2^{O(\sqrt{\log n})}$ rounds. Then, this minimum-ID node locally simulates the ball-carving process, as follows: each time, it picks another unprocessed node in its component, and then carves the ball around it similar to the sequential ball-carving process explained above. Notice that this ball can potentially go out of $G_{j}$. However, the ball will grow at most $n^{1/\lambda}\cdot \log n$ hops away from its center. Hence, the ball carving processes of two different components of $G_j$ never reach each other, as the components are more than $d > 2n^{1/\lambda}\cdot \log n$ hops apart. Also, notice that the minimum-ID node is doing this process locally, once it has gathered the relevant information, and thus this computation does not consume any further communication. Once the minimum-ID node has computed the newly carved $A_i$-balls centered at the nodes of its component, it informs the related nodes of their status: whether they are in $A_1$, discarded from the $A_1$ due to falling on the boundary, or remaining unprocessed. This finishes the description of the $j^{th}$ phase. We then move to the next phase.  
%
%
%
 %
%
%
%
\end{proof}

\subsection{The Shattering Lemma}\label{app:shattering}

\begin{proof}[Proof of \Cref{Shattering}]
The proof is similar to \cite{barenboim2012locality}; thus we provide only a brief sketch. The existence of such a connected set $U$ in (P1) would imply that $H[B]$ contains a tree on $\log_{\Delta} n$ nodes. There are at most $4^{\log_{\Delta} n}$ different such tree topologies and each can be embedded into $H$ in less than $n\cdot \Delta^{(4c_2+2)\log_{\Delta} n}$ ways. Moreover, the probability that a particular tree occurs in $H[B]$ is at most $\Delta^{-c_1 \log_{\Delta} n}$, since all the nodes have distance at least $2c_2+1$, which means that all the events depend on disjoint sets of coin tosses. 
%
A union bound over all trees thus lets us conclude that such a set $U$ exists with probability at most $4^{\log_{\Delta} n}\cdot n \cdot \Delta^{(4c_2+2)\log_{\Delta} n}\cdot  \Delta^{-c_1 \log_{\Delta}n}\leq n^{-c_3}$. This proves (P1).  

To prove (P2), we observe that a connected component $S$ in $G[B]$ of size $ \log_{\Delta} n \cdot \Delta^{2c_2}$ implies the existence of a connected set $U$ in $H[B]$ of size $\log_{\Delta} n$. To that end, we greedily add nodes from $S$ to $U$ one by one, each time discarding all (at most $\Delta^{2c_2}$ many) nodes within $\paren{2c_2}$-hops of the added node. It follows that every connected component of $G[B]$ has size at most $O\paren{\log_{\Delta} n \cdot \Delta^{2c_2} }$, with probability at least $1- n^{-c_3}$.
%

To prove (P3), we first compute a $(4c_2 +3 , \Theta(\log\log n))$-ruling set of each connected component $C_B$ of $B$ (with regard to the distances in $G$), in $O(\log\log n)$ rounds, using the algorithm of Schneider, Elkin and Wattenhofer\cite{schneider2013symmetry}. See also\cite[Table 4]{barenboim2012locality}. Recall that an $(\alpha, \beta)$-ruling set for $C_B$ means a set $S_B$ of vertices where each two have distance at least $\alpha$, while for each node in $C_B$, there is at least one node in $S_B$ within $\beta$ hops. Then, each node joins the cluster of its nearest ruling set node. We contract the clusters into new nodes, and connect two new nodes if their clusters contain adjacent nodes. Hence, we obtain a new graph where each component has at most $O(\log n)$ nodes, by property (P1). We then compute a $(\lambda, \log^{1/\lambda} n \cdot \log\log n)$ network decomposition of each of these $O(\log n)$-size graphs, independently and all in parallel, using the algorithm described in \Cref{NWD}. Notice that when invoking \Cref{NWD}, we are now working on graphs of size $O(\log n)$, which means the network decomposition that we obtain has parameters $(\lambda, O(\log^{1/\lambda} n \cdot\log\log n))$. At the end, we extend this decomposition to the original graph on nodes of (this connected component of) $B$, where each vertex of $B$ belongs to the block of the network decomposition where its contracted cluster is. This increases the radius of each of the blocks by the radius of the contracted cluster, which is at most an $O(\log\log n)$-factor, hence leading to a $\big(\lambda, O(\log^{1/\lambda} n\cdot  \log^2\log n)\big)$ network decomposition.
\end{proof}

\section{Missing Details of \Cref{sec:GeneralLLLAlgo}: Bootstrapping}
\label{app:GeneralLLLAlgo}
\begin{proof}[Proof of \Cref{lemma:SpecialSpeedup}]
Consider the randomized algorithm $\mathcal{A}_{n^*}$ that solves some LCL problem $\mathcal{P}$ on $n^*$-node graphs with complexity $O(\log^{1/4} n^*)$. We now bootstrap $\mathcal{A}_{n^*}$ using an approach similar to the proof of \Cref{RANDLLL}: In particular, we shall run $\mathcal{A}_{n^{*}}$ on our full graph of $n$ nodes, where we set $n^*=\log n$, while $\mathcal{A}_{n^*}$ is still told that the network size is $n^{*}$. This runs in $O(\log^{1/4} \log n)$ rounds. The probability of each local bad event---i.e., a violation of one of the conditions of $\mathcal{P}$---is at most $p'=1/n^*=1/\log n$. On the other hand, each two of these local bad events that are further than  $O(\log^{1/4} \log n)$ hops apart rely on disjoint random variables in the execution of $\mathcal{A}_{n^*}$. Hence, this new LLL system has dependency degree at most $d'=d^{O(\log^{1/4} \log n)} = 2^{O(\sqrt{\log\log n})}$. Thus, this system satisfies the polynomial LLL criterion with a value of $\lambda=\Theta(\sqrt{\log\log n})$, because $p'(ed')^{\lambda}<1$. Hence, we can solve it using the algorithm of \Cref{thm:baseLLL} in $O((d')^2) + \lambda \cdot (\log n)^{1/\lambda} \cdot 2^{O(\sqrt{\log \log n})}=2^{O(\sqrt{\log\log n})}$ rounds of the new LLL system. Each of these rounds can be performed in $O(\log^{1/4} n^*)=O(\log^{1/4}\log n)$ rounds of the base graph, and thus the overall round complexity of the new algorithm $\mathcal{A}'_{n}$ is $2^{O(\sqrt{\log\log n})}$.
\end{proof}


\begin{proof}[Proof of \Cref{crl:loglogTologStar}]
Consider the randomized algorithm $\mathcal{A}_{n^*}$ that solves the LCL problem $\mathcal{P}$ on $n^*$-node graphs in $o(\log\log n^*)$ rounds. We bootstrap $\mathcal{A}_{n^*}$ using an approach similar to the proof of \Cref{RANDLLL}. 
In particular, we shall run $\mathcal{A}_{n^{*}}$ on our full graph of $n$ nodes, while $\mathcal{A}_{n^*}$ is still told that the network size is $n^{*}$, for a sufficiently large constant value of $n^*$. This runs in $T=o(\log\log n^*)$ rounds. The probability of each local bad event---i.e., a violation of one of the conditions of $\mathcal{P}$---is at most $p'=1/n^*$. On the other hand, each two of these local bad events that are further than  $2T+1=o(\log\log n^*)$ hops apart rely on disjoint random variables in the execution of $\mathcal{A}_{n^*}$. Hence, this new LLL system has dependency degree at most $d'=d^{2T+1}$. Thus, this system satisfies the polynomial LLL criterion with $\lambda\geq (d')^2+1$. That is because $$p(ed')^{d'+1} \leq \frac{1}{n^*} \big(e d^{o(\log\log n^*)}\big)^{d^{o(\log\log n^*)}+1} < \frac{1}{n^*} \big(\sqrt{\log n^*} \big)^{\sqrt{\log n^*}} <1,$$ where the penultimate inequality uses that $d=O(1)$.

On the other hand, we can easily compute a $((d')^2+1, 0)$ network decomposition of the square graph $G^2_{\mathcal{X}}$ of this new LLL's dependency graph, simply by taking a $((d')^2+1)$-coloring of it. Notice that this coloring can be computed in $O(\log^* n)$ time, using the deterministic distributed coloring algorithm \cite{fraigniaud2016local}. Then, we apply the deterministic LLL algorithm of \Cref{DETLLL} on top of this network decomposition, with $\lambda=(d')^2+1$ and $\gamma=0$. The algorithm of \Cref{DETLLL} then runs in $O((d')^2)$ rounds, and solves this LLL, hence providing a solution for the LCL problem $\mathcal{P}$.  Overall, we get a deterministic algorithm with complexity $O(\log^* n)$ for solving the LCL problem $\mathcal{P}$ on bounded degree graphs.
\end{proof}

\section{Missing Details of \Cref{sec:DefectiveColoring}: Defective Coloring}
\label{app:DefectiveColoring}

\begin{proof}[Proof of \Cref{defective-col}]
If $f \geq \Delta$, any assignment of colors to nodes is an $f$-defective coloring. If $f=O(1)$, a proper $O(\Delta)$-coloring is a $O(\Delta/f)$-coloring with defect $0 \leq f$. In this case, we can find such a coloring by running the algorithm of Barenboim et. al. \cite{barenboim2012locality} in $2^{O(\sqrt{\log \log n})}$ rounds. We thus assume in the following that $f < \Delta$ and $f=\omega(1)$.

\paragraph{Parameters}
Let $\Delta_0=\Delta$, and for $i \geq 1$, set $\Delta_i=\log^5 \Delta_{i-1}$, $\eps_i = \log^2 \Delta_{i-1} / \sqrt{\Delta_i} = \log^2 \Delta_{i-1}/\log^{2.5} \Delta_{i-1} = \log^{-0.5} \Delta_{i-1}$, and $k_i=(1+\eps_i)\Delta_{i-1}/\Delta_i$. Moreover, let $t$ be such that $\Delta_t = \omega(f)$ and $\log^5 \Delta_{t}=O(f)$, and set $\Delta_{t+1}=f$, $\eps_{t+1}=\log^2 \Delta_t / \sqrt{f}$, as well as $k_{t+1}=(1+\eps_{t+1})\Delta_t/f$. 

\paragraph{Iterated Bucketing}
We run the bucketing algorithm of \Cref{Bucketing} with $\Delta' \gets \Delta_1$, $\eps \gets \eps_1$, $k \gets k_1$ for at most $3$ iterations (in each iteration, the algorithm is applied to each of the buckets from the previous iteration), until we have reached $\Delta_i = O(\log^{1/10} \log n)$. Then we switch to the direct LLL algorithm of \Cref{RANDLLL}, run for bucketing, and perform it until $i=t$. Then, for $i=t+1$, we apply the LLL algorithm of \Cref{RANDLLL} one last time with $\Delta'\gets \Delta_{t+1}=f$.    

\paragraph{Analysis}
We first observe that $t=O(\log^*\Delta)$. The overall running time thus is $O(t)\cdot 2^{O(\sqrt{\log \log n})}=2^{O(\sqrt{\log \log n})}$, since each of the $t+1$ bucketing iterations takes at most $2^{O(\sqrt{\log \log n})}$ rounds by \Cref{Bucketing,RANDLLL}. Notice that all the buckets on the same level (that is, in the same iteration) are treated independently, in parallel.  

In the end, after these $t+1$ iterations, the resulting $f$-defective coloring has at most $(2k_1) \cdot (2k_2) \cdot (2k_3) \prod_{i=4}^{t+1}k_t \leq 4 \prod_{i=1}^{t+1} (1 +\eps_i) \frac{\Delta_{i-1}}{\Delta_i}\leq 4 (1+ O(\eps_{t+1}))\frac{\Delta}{f}= 8(1+O(f^{-1/10}))\frac{\Delta}{f}$ colors, using that $\eps_i$ is exponentially increasing in $i$ and that $\log^2 \Delta_t = O(f^{2/5})$.
\end{proof}

\begin{proof}[Proof of \Cref{Bucketing}]
We break the $\Delta \geq \Omega(\log^{1/10} \log n)$ range into two sub-ranges, based on whether $\Delta \geq 2^{\Omega({\log^{1/4}\log n})}$ or not. We present the proof for each of these cases separately.

\paragraph{Case 1, $\Delta \geq 2^{\Omega({\log^{1/4}\log n})}$} We assign each node to one of the first $k$ buckets uniformly at random. Then, for each node that has more than $\Delta'$ neighbors in its bucket, we remove it from its bucket and put it into $B$. Note that even though we might have removed a node from its bucket (and added it to $B$) in this way, it still counts as neighbor for other nodes in its bucket. By construction, each of the $k$ buckets has degree at most $\Delta'$. However, we have a set $B$ of nodes not assigned to any bucket. We now show how to perform another bucketing step of $B$ into $k$ additional buckets, by setting up and solving another bucketing LLL. 

By the above observations, a node is put in $B$ with probability at most $e^{-\Omega(\log^2 \Delta)}$. Moreover, the events $1(v \in B)$ depend only on the $1$-hop neighborhood.
By \Cref{Shattering} (P3), we can compute a $(\sqrt{\log \log n}, 2^{O(\sqrt{\log\log n})}\cdot \log^2 \log n)$ network decomposition of (each connected component of) $G[B]$ in $2^{O(\sqrt{\log \log n})}$ rounds, setting $\lambda = \sqrt{\log \log n}$. We now set up a bucketing LLL (in fact, one for each of the connected components of $G[B]$), and invoke the deterministic algorithm of \Cref{DETLLL} on top of this network decomposition. Notice that the criterion $p(ed)^{\lambda}<1$ is satisfied for $\lambda= \sqrt{\log \log n}$, because $p \leq \Delta^{-\Omega(\log^3 \Delta)}$ and $\Delta=2^{\Omega(\log^{1/4} \log n)}$. Hence, the algorithm of \Cref{DETLLL} runs in 
$2^{O(\sqrt{\log\log n})}$ rounds, and computes a bucketing for all the nodes in $G[B]$ into $k$ additional buckets.
%
%
\medskip

\paragraph{Case 2, $\Delta\in [\Omega(\log^{1/10}\log n), 2^{O(\log^{1/4} \log n)}]$}
We first devise a randomized algorithm $\mathcal{A}_{n^*}$ with complexity $O(\log^{1/4} n^*)$ that computes a bucketing of any $n^*$-node graph into $2k$ buckets, each with maximum degree $\Delta'$. Then, we bootstrap this algorithm, using \Cref{lemma:SpecialSpeedup}, to turn it into another bucketing algorithm $\mathcal{A}'_{n}$ that runs in $2^{O(\sqrt{\log \log n})}$ on any $n$-node graph with $\Delta \leq 2^{O(\log^{1/4} \log n)}$.

Algorithm  $\mathcal{A}_{n^{*}}$ performs a simple bucketing by putting each node in one of the first $k$ buckets, chosen uniformly at random. For each node that has more than $\Delta'$ neighbors in its bucket, we remove it from its bucket and put it into a set $B$ of bad nodes. By \Cref{Shattering} (P3), we can compute a $(8, O(\log^{1/4} n^*))$ network decomposition of each connected component of the subgraph induced by $B$, in $O(\log^{1/4} n^*)$ rounds. We can then invoke the deterministic LLL algorithm of \Cref{DETLLL} to compute a bucketing of these bad nodes of $B$ into the second $k$ buckets, in no more than $O(\log^{1/4} n^*)$ rounds, by setting $\lambda=8$ in \Cref{DETLLL}. This completes the description of the bucketing algorithm $\mathcal{A}_{n^*}$ with complexity $O(\log^{1/4} n^*)$.

We now bootstrap $\mathcal{A}_{n^*}$ using \Cref{lemma:SpecialSpeedup} to obtain a bucketing algorithm $\mathcal{A}'_{n}$ that runs in $2^{O(\sqrt{\log \log n})}$ rounds, on any $n$-node graph with maximum degree at most $2^{\Theta(\log^{1/4} \log n)}$.
\end{proof}

\section{Missing Details of \Cref{sec:FrugalColoring}: Frugal Coloring}
\label{app:FrugalColoring}

\begin{proof}[Proof of \Cref{frugal-iteration}]
We let $C:=20 \cdot  \Delta' \cdot \Delta^{1/\beta}$. Every node $v\in V'$ picks $x$ tentative colors $c_j(v)$ for $j \in [x]$ uniformly at random from $x$ disjoint palettes, each of size $C$. 
Then node $v$ gets permanently colored with the first color that does not lead to a conflict, if there is at least one such color. 

For the sake of analysis, we think of these $x$ samplings happening sequentially in steps $1 \leq j \leq x$. 
%
We introduce two type of events. An event $$M_j(v):=\left\{\exists u \in N(v,V') \colon c_j(u)=c_j(v)\right\}$$ if $v$ has a monochromatic incident edge in step $j$, in other words, if there is a node in $V'$ adjacent to $v$ that picks the same color as $v$ in step $j$. An event $$F_j(v):=\left\{ \exists u \in N(v), u_1, \dotsc, u_{\beta}\in N\paren{u, V' \setminus \{v\}} \colon c_j\paren{u_i}=c_j(v) \text{ for all } i \in [\beta] \right\}$$ if $v$ is involved in a non-frugal neighborhood of some node $u\in V$ with respect to its own $j^{th}$ tentative color choice. With $U_0=V'$, we let $M_j$ and $F_j$ denote the set of nodes $v \in U_{j-1}$ for which $M_j(v)$ and $F_j(v)$, respectively, holds. We discard the tentative colors of all nodes in $U_j:=M_j \cup F_j$ and make the tentative color of nodes in $K_j:=U_{j-1}\setminus U_j$ permanent. 

We note that even though a node $v\in V'$ might have a permanent color, and thus be part of $K_j$ for some $j$, it still participates in the color sampling for $j'>j$, and can lead to conflicts with other nodes in $u\in U_{j'}$, forcing $u$ to uncolor itself, while $v$ stays colored permanently.

\paragraph{Validity of the Coloring}
Let $K:=\bigcup_{j=1}^x K_j$ be the set of nodes that succeed in finding a permanent color, and let $U:=U_x=V' \setminus K$ denote the set of all nodes that remain uncolored. 
It is easy to see that the coloring of the nodes in $K$ uses at most $x\cdot C$ colors and is a partial $\beta$-frugal coloring. 

\paragraph{Properties (i) and (ii)}
We bound the probability of a single node in $V'$ being uncolored in a single step, show that every node has many different colors in its neighborhood in each step, and then conclude that the number of uncolored neighbors of a node must decrease in every step by arguing about each color separately. Finally, we combine these results about a single step to derive a bound on the number of uncolored neighbors of a node after all $x$ trials. 
%
%

\paragraph{Probability of Being Uncolored}
For $j \in [x]$ and $v \in U_{j-1}$, we have $$Pr[v \in U_j]\leq Pr[M_j(v)]+Pr[F_j(v)]\leq \frac{\Delta'}{C}+ \Delta \cdot {{\Delta'-1}\choose \beta} \cdot \frac{1}{C^{\beta}} \leq \frac{\Delta'}{C} + \frac{\Delta \cdot (\Delta')^{\beta}}{C^{\beta}} \leq 1/10.$$ Moreover, $Pr[v \in U] \leq \left(1/10\right)^x$, as different steps use different palettes. This proves (i). 

\paragraph{Number of Different Colors in Neighborhood in One Step}
We show that with large probability a node $v\in V$ has many different tentative colors in the neighborhood $N\paren{v,U_{j-1}}$ of size $d$.
For an arbitrary ordering $u_1, \dotsc, u_{d}$ of the neighbors of $v$, we introduce random variables $X_i$ which indicate whether node $u_i$ has a color different from all the colors of nodes $u_1, \dotsc, u_{i-1}$ in step $j$. Then $X:=\sum_{i=1}^{d} X_i$ is the total number of different colors in the neighborhood $N\paren{v,U_{j-1}}$ of $v$. We have $\mathbb{E}[X]\geq d- \frac{1}{C}\sum_{i=1}^{d} (i-1) \geq 
\left(1-\frac{1}{2\cdot 20}\right)d$. 
A Chernoff bound, applicable due to negative correlation of $X_i$, yields
$Pr\left[X < (1-1/20) d\right] \leq Pr\left[X \leq (1-1/40)\mathbb{E}[X]\right]\leq e^{-\mathbb{E}[X]/4800}\leq e^{-d/5000}.$
 
\paragraph{Degree in Uncolored Graph in One Step}
We first observe that the events $u\in U_j$ and $w\in U_j$ for nodes with $c_j(u)\neq c_j(w)$ are negatively correlated (conditioned on their colors). Intuitively speaking, when there is a conflict with one color, it is unlikelier or as unlikely that there is a conflict with another color. For nodes of the same color, however, these events might be positively correlated. 

For the moment, we suppose that $v\in V$ has $t$ many different tentative colors in its neighborhood in $U_{j-1}$ of size $d(v, U_{j-1})=d$ and let $u_1, \dotsc, u_t \in U_{j-1}$  be nodes having the respective colors. By the above observations, each of the events $u_i \in U_j$ has probability at most $1/10$ and these events are negatively correlated. It follows by a Chernoff bound that the probability that more than $(3t)/20$ of these nodes are uncolored is at most $e^{-t/120}$. 

Having more than $d/5$ uncolored nodes in the neighborhood $N\paren{v,U_{j-1}}$ means having less than $\left(1-1/5\right)d$ colored nodes in the neighborhood $N\paren{v,U_{j-1}}$. That in particular means having less than $\left(1-1/5\right)d$ colored nodes among $u_1, \dotsc, u_t$, which is in other words having more than $t-(1-1/5)d$ uncolored  among $u_1, \dotsc, u_t$. Thus, since $t-(1-1/5)d\geq (3d)/20$ for $t \geq \left(1-1/20\right)d$, we have 
\begin{equation*}\begin{aligned}Pr\left[d\paren{v,U_{j-1}} >d/5\right] &= \sum_{t=1}^{d}Pr\left[d\paren{v,U_{j-1}} >d/5 \mid X=t \right]\cdot Pr[X=t]
 \\ &\leq e^{-d/5000}+ \sum_{t=\paren{1-1/20}d}^{d} Pr\left[d\paren{v,U_{j-1}} > d/5 \mid X=t \right]
\\ &\leq e^{-d/5000} + \sum_{t=\paren{1-1/20}d}^{d}  e^{-t/120}
\leq e^{-d/5000} + (d/20)\cdot e^{-\frac{\paren{1-1/20}d}{120}}
\\ &\leq e^{-d/10000}.
\end{aligned}\end{equation*} 

\paragraph{Degree in Uncolored Graph After $x$ Trials}
It follows by a union bound over all $x$ steps that 
$Pr\left[d\paren{v,U_x} > 5^{-x} \cdot\Delta'\right] \leq \sum_{i=1}^x e^{-(1/10000)\cdot 5^{-i} \cdot \Delta'}=e^{-\Omega(5^{-x} \cdot \Delta')}$,
where the last step follows from the fact that $e^{-5^{-i}}$ is (doubly-)exponentially increasing in $i$. 
%
\end{proof}

\begin{proof}[Proof of \Cref{lemma:frugal-iteration-sampling}]
We handle the problem in four cases, depending on the range of $\Delta$.

\medskip
\paragraph{Case 1, $\Delta=\omega(\log^2 n)$} We perform a sampling as described in \Cref{frugal-iteration}. Then, by \Cref{frugal-iteration} (ii), the probability of a node having more than $\Delta''=5^{-x} \Delta'$ uncolored neighbors is at most $p=e^{-\omega(\log n)}$. A simple union bound over all nodes shows that, with high probability, there is no such node whose base degree remains high.

\medskip
\paragraph{Case 2, $\Delta \in [\omega(\log^2 \log n), O(\log^2 n)]$} We perform a sampling as described in \Cref{frugal-iteration}. This gives us a partial $\beta$-frugal $(20\cdot x \cdot \Delta' \cdot \Delta^{1/\beta})$-coloring, assigning colors to a subset $W \subseteq V'$. Let $D \subseteq V$ denote the set of nodes in $V$ that have degree larger than $\Delta''$ into the uncolored set $U=V' \setminus W$, and let $B=N(D) \cap U$ be all the neighbors in $U$ of such nodes in $D$. We will show how to partially color (some) nodes in $B$ with additional $20 \cdot x \cdot \Delta' \cdot \Delta^{1/\beta}$ colors such that the base-graph degree into the uncolored set $B'$ in $B$ drops to $\Delta''$. This is a partial $\beta$-frugal $(40\cdot x \cdot \Delta' \cdot \Delta^{1/\beta})$-coloring such that no node in $V$ has more than $\Delta''$ neighbors in the uncolored set $(U\setminus B )\cup B'$, as desired.




Consider the square graph $G^2[B]$, i.e., the graph on vertices of $B$ where each two $B$-nodes whose $G$-distance is at most $2$ are connected. 
By \Cref{frugal-iteration} (ii), the probability of a node $u$ being in $D$ is $e^{-\Omega(5^{-x} \cdot \Delta')}$. Hence, the probability of $v$ being in $B$ is at most $\Delta \cdot e^{-\Omega(5^{-x} \cdot \Delta')}= e^{-\Omega(\sqrt{\Delta})}$, by a union bound over all neighbors $u \in N(v)$ of $v$, due to the assumption that $5^{-x} \cdot \Delta'=\Omega(\sqrt{ \Delta})$. Moreover, only the events $u \in B$ for nodes with distance at most $3$ in $G$, and hence distance at most $6$ in $G^2$, are dependent. \Cref{Shattering} (P3) thus shows that the connected components of $G^2[B]$  w.h.p. admit a $(\lambda, O(\log^{1/\lambda}\cdot \log^2\log n))$ network decomposition, which we can compute in $\lambda \cdot \log^{1/\lambda} n \cdot 2^{O(\sqrt{\log \log n})}$ rounds, deterministically.  

We now set up a polynomial LLL for a ``progress-guaranteeing'' partial frugal coloring on the (uncolored) set $B$ (with all the base graph nodes $V$), as discussed in \Cref{LLLPartialFrugal}. In fact, we set up one independent such LLL for each component of $G^{2}[B]$. Notice that the colorings of different components do not interfere with each other, as each $V$-node has neighbors in at most one of these components. This LLL satisfies the stronger condition $p(ed)^{4\lambda}$ of \Cref{DETLLL} for $\lambda=\sqrt{\log \log n}$. We can thus apply the deterministic algorithm of \Cref{DETLLL} on top of this network decomposition, which runs in $O(\lambda \cdot \log^{1/\lambda}\cdot \log^2 \log n)$ rounds. Overall, this takes $\lambda \cdot \log^{1/\lambda} n \cdot 2^{O(\sqrt{\log \log n})}+ O(\lambda \cdot \log^{1/\lambda}\cdot \log^2 \log n)=2^{O(\sqrt{\log \log n})}$ rounds, and gives us a partial coloring of nodes in $B$ such that the base-graph degree to the nodes that remain uncolored has dropped to $\Delta''$.

\medskip
\paragraph{Case 3, $\Delta \in [\omega(\log^{1/5} \log n), O(\log^2 \log n)]$ } We first devise an algorithm $\mathcal{A}_{n^{*}}$ that performs the desired partial frugal coloring in $O(\log^{1/4} n^*)$ rounds, in $n^{*}$-node graphs. Then we use \Cref{lemma:SpecialSpeedup} to speed up this partial frugal coloring algorithm to run in $2^{O(\sqrt{\log\log n})}$ rounds. 

The algorithm $\mathcal{A}_{n^{*}}$ is similar to the process we described in case $2$: it first performs a sampling according to \Cref{frugal-iteration}, then defines bad nodes $B$ for this sampling similar to before. The only difference is that, when solving the LLL of each of the components of $G^2[B]$, we may not have the desired polynomial LLL criterion satisfied for $\lambda=\Omega(\sqrt{\log\log n^*})$. Still, the condition is satisfied for any desirable large constant $\lambda=O(1)$. Hence, the deterministic algorithm of \Cref{DETLLL} solves these remaining components in no more than $O(\log^{1/4} n^*)$, with probability $1-1/\poly(n^{*})$. This is the desired algorithm $\mathcal{A}_{n^{*}}$ for partial frugal coloring.

Now, we invoke \Cref{lemma:SpecialSpeedup} to speed up this partial frugal coloring algorithm $\mathcal{A}_{n^{*}}$ to run in $2^{O(\sqrt{\log\log n})}$ rounds, with high probability, on any $n$-node graph with maximum degree at most $\Delta=O(\log^2 \log n)$. 
Notice that we are able to do this because the maximum degree $\Delta=O(\log^2 \log n)$ is (even far) below the requirement $\Delta\leq 2^{\Theta(\log^{1/4}\log n)}$ of \Cref{lemma:SpecialSpeedup}. Thus, we get an algorithm $A_{n}$ for partial frugal coloring, that w.h.p. in $2^{O(\sqrt{\log\log n})}$ rounds colors a subset of $V'$ such that the base-degree to the nodes that remain uncolored is at most $\Delta''$.
\newpage
\paragraph{Case 4, $\Delta \in O(\log^{1/5} \log n)$} Here, we can directly apply the LLL algorithm of \Cref{RANDLLL} to the ``progress-guaranteeing'' LLL. This yields, w.h.p., a drop in the degree to $\Delta''$, in $2^{O(\sqrt{\log\log n})}$ rounds.
\end{proof}

\begin{proof}[Proof of \Cref{lemma:whole-iterations}]
We first set the parameters, then formalize the exact meaning of iterating the sampling process of \Cref{lemma:frugal-iteration-sampling}, and finally analyze the number of colors as well as rounds used.

\paragraph{Parameters}
We let $x_0=1$, $x_{i+1}:= \left(5/4\right)^{x_i}$, $\Delta_0=\Delta'$, and $\Delta_{i+1} = 5^{-x_i}\cdot \Delta_i$ for $0\leq i \leq t$ for a $t=O(\log^* \Delta)$ such that $\Delta_{t+1}=O(\sqrt{\Delta})$ for the first time (that is, $\Delta_t = \omega(\sqrt{\Delta})$). 

\paragraph{Iterated Sampling}
Iteratively, for $0 \leq i\leq t$, we apply the sampling algorithm of \Cref{lemma:frugal-iteration-sampling} with $V' \mapsto V_i$, $\Delta' \mapsto \Delta_i$, and $x \mapsto x_i$ to obtain a partial $\beta$-frugal coloring with $C_i:=40 \cdot x_i \cdot \Delta_i \cdot \Delta^{\frac{1}{\beta}}$ many new colors, leaving a set $V_{i+1} \subseteq V_i$ with base-graph degree at most $\Delta_{i+1}$ uncolored.



\paragraph{Analysis}
Intuitively, as we know that the number of colors needed to find a $\beta$-frugal coloring decreases with $\Delta_i$, we can afford to use more and more disjoint palettes when $\Delta_i$ is small. For $x_i \cdot \Delta_i = \Delta$, this would mean that we use the same number of colors in each iteration $i$. However, this would lead to a total number of $t \cdot 40 \cdot \Delta^{1 + 1/\beta}$ colors over all the $t$ iterations. Instead, we ensure that $x_{i+1} \cdot \Delta_{i+1}$ is at least a constant factor smaller than $x_{i} \cdot \Delta_i$, which guarantees that the total number of colors used behaves like a geometric series. Indeed,
$
\sum_{i=0}^t 40 x_i \cdot \Delta_i \cdot \Delta^{1/\beta}= 40\cdot\Delta^{1+\frac{1}{\beta}}\cdot\sum_{i=0}^t 2^{-\sum_{j=0}^i x_j} \leq 40 \cdot \Delta^{1+\frac{1}{\beta}}\sum_{i=0}^t 2^{-i}\leq 80 \cdot \Delta^{1+\frac{1}{\beta}}.$

Finally, observe that each of the $O(\log^* \Delta)$ iterations takes $2^{O(\sqrt{\log \log n})}$ rounds by \Cref{lemma:frugal-iteration-sampling}, which proves that the overall round complexity is $2^{O(\sqrt{\log \log n})}$.
\end{proof}

\begin{proof}[Proof of \Cref{lemma:frugal-color-good}]
We handle the problem in four cases, depending on the range of the values of $\Delta$, similar to the proof of \Cref{lemma:frugal-iteration-sampling}.

\medskip

\paragraph{Case 1, $\Delta=\omega(\log^2 n)$} We perform one sampling step as described in \Cref{frugal-iteration} with $x=\Delta/\Delta'=\Omega(\sqrt{\Delta})$. Then, by \Cref{frugal-iteration} (ii), the probability of a node remaining uncolored is at most $p=e^{-\omega(\log n)}$. A simple union bound over all nodes shows that, with high probability, no node remains uncolored.

\medskip
\paragraph{Case 2, $\Delta \in [\omega(\log^2 \log n), O(\log^2 n)]$} We perform one sampling step as described in \Cref{frugal-iteration} with $x=\Delta/\Delta'=\Omega(\sqrt{\Delta})$. This gives a partial $\beta$-frugal coloring with $x\cdot 20 \cdot \Delta' \cdot \Delta^{1/\beta}=20 \cdot \Delta^{1+1/\beta}$ colors. Let $V'' \subseteq V'$ be the set of nodes that remain uncolored. The probability of a node being in $V''$ is at most $10^{-x}=e^{-\Omega(\sqrt{\Delta})}$, by \Cref{frugal-iteration} (i). Moreover, uncoloring for nodes with distance at least 3 in $G$, and hence at least $5$ in $G^2$, are independent. Thus, \Cref{Shattering} (P3) implies that the connected components of $G^2[V'']$ w.h.p. admit a $(\lambda, O(\log^{1/\lambda}n \cdot \log^2 \log n))$ network decomposition, which can be computed in $\lambda \cdot \log^{1/\lambda}n \cdot 2^{O(\sqrt{\log \log n})}$ rounds, deterministically. 

We set up LLLs with $x=\Delta/\Delta'$ for completing a partial frugal coloring, one for each connected component of $G^2[V'']$. Notice that the colorings of different components do not interfere with each other, as each $V$-node has neighbors in at most one of these components.  Setting $\lambda=\sqrt{\log \log n}$, we have $p (ed)^{\lambda} = (10)^{-\sqrt{\Delta}}(O(\Delta))^{2\sqrt{\log \log n}}< 1$. 
Thus, the even stronger condition of \Cref{DETLLL} for $\lambda = \sqrt{\log \log n}$ is satisfied, which lets us find a solution for the LLL, and hence a completion of the $\beta$-frugal coloring, in additional $\lambda \cdot \log^{1/\lambda} n \cdot \log^2 \log n$ rounds, on each of the connected components of $G^2[V'']$ in parallel.
Overall, this takes $\lambda \cdot \log^{1/\lambda}n \cdot 2^{O(\sqrt{\log\log n})} +\lambda \cdot \log^{1/\lambda} n \cdot \log^2 \log n = 2^{O(\sqrt{\log \log n})}$ rounds.  
\paragraph{Case 3, $\Delta \in[\omega(\log^{1/5} \log n), O(\log^2 \log n)]$} We first devise an algorithm $\mathcal{A}_{n^{*}}$ that completes the given partial frugal coloring in $\Theta(\log^{1/4} n^*)$ rounds, in $n^{*}$-node graphs. Then, we use \Cref{lemma:SpecialSpeedup} to speed up this coloring completion algorithm to run in $2^{O(\sqrt{\log\log n})}$ rounds. 

The algorithm $\mathcal{A}_{n^{*}}$ is similar to the process we described in case $2$: it first performs a sampling according to \Cref{frugal-iteration} with $x=\Delta/\Delta'=\Omega(\sqrt{\Delta})$, then defines bad nodes $B$ for nodes that remain uncolored, and handles each of connected components of $G^2[B]$ with a new LLL. The only difference is that, when solving the LLL of each of the components of $G^2[B]$, we may not have the desired polynomial LLL criterion satisfied for $\lambda=\Omega(\sqrt{\log\log n^*})$. Still, the condition is satisfied for any desirably large constant $\lambda$. Hence, the deterministic algorithm of \Cref{DETLLL} can solve these remaining components in at most $\Theta(\log^{1/4} n^*)$ rounds, with local correctness probability at least $1-1/\poly(n^{*})$. This is the desired algorithm $\mathcal{A}_{n^{*}}$ for the completion of the coloring.

Now, we invoke \Cref{lemma:SpecialSpeedup} to speed up this frugal coloring completion algorithm $\mathcal{A}_{n^{*}}$ to run in $2^{O(\sqrt{\log\log n})}$ rounds, with high probability, on any $n$-node graph with maximum degree at most $\Delta=O(\log^2 \log n)$. Notice that we are able to do this because the maximum degree $\Delta=O(\log^2 \log n)$ is (even far) below the requirement $\Delta\leq 2^{\Theta(\log^{1/4}\log n)}$ of \Cref{lemma:SpecialSpeedup}. Thus, we get an algorithm $A_{n}$ for completing the partial frugal coloring that, in $2^{O(\sqrt{\log\log n})}$ rounds, colors all the remaining uncolored nodes, with high probability.

\vspace{-2pt}
\paragraph{Case 4, $\Delta \in O(\log^{1/5} \log n)$} Here, we can directly apply the LLL algorithm of \Cref{RANDLLL} to the LLL for completing a partial frugal coloring. This gives us an algorithm that w.h.p. completes the given partial frugal coloring in $2^{O(\sqrt{\log\log n})}$ rounds.
\end{proof}
\vspace{-20pt}
\section{Missing Details of \Cref{app:listColoring}: List Vertex-Coloring}
\label{app:savingLLLpruning}
\subsection*{Solving the LLL for $2$-factor Pruning of Color Lists}
\vspace{-5pt}
For $L \leq O((\log\log n)^{1/10})$, we can directly apply the LLL algorithm of \Cref{RANDLLL} to solve the above $2$-factor pruning in $2^{O(\sqrt{\log\log n})}$ rounds. In the following, we discuss how we handle the remaining case $L\in [\Omega(\log^{1/10}\log n), O(\log^2 n)]$. We break this range into two cases, depending on whether $L\geq \Omega(\log^4\log n)$ or not. A key part in both will be a somewhat gradual sampling of which colors to retain in the list. To perform that sampling with an appropriate speed (to be made precise), we will use defective colorings, as we discuss next. 

\vspace{-2pt}
\paragraph{Color-Choice Graph} Consider a graph $H$ where we include one vertex $(v, q)$ for each color $q\in L_v$ of each node $v$. Two vertices $(v, q)$ and $(u, q')$ are connected if and only if either (1) $v=u$, or (2) $v$ and $u$ are adjacent and $q=q'$. 

\vspace{-2pt}
\paragraph{Defective Coloring of the Color-Choice Graph} Notice that the color-choice graph $H$ has maximum degree at most $L+L/C\leq 2L$. We compute a defective coloring $\chi$ of $H$ with defect $f=L/(2\log^2 L)$ and $O((\frac{2L}{f})^2) = O(\log^4 L)$ colors, in $O(\log^* n)$ rounds, using the deterministic algorithm of Kuhn\cite{kuhn2009weak}. We use this defective coloring mainly to schedule which colors $(v, q)$ are sampled to be kept in the list. 

\vspace{-2pt}
\paragraph{Sampling the Colors, in $O(\log^4 L)$ Phases} We have $K_0\log^4 L$ phases, for some constant $K_0$, one per color class of the schedule-color $\chi$. During each phase $i$, we sample each of the colors $(v, q) \in H$ who has $\chi$-color $i$, with probability $1/2$, for inclusion in $L'_v$. At the end of the phase, we check two properties, and potentially freeze some of the unset color-choices in $H$, meaning that we will not sample these, and we defer the decision on them to some later process. This freezing is done as follows: If for a node $v$, we had $z_v \geq L/(16K_0 \log^6 L)$ many of its colors $(v, q)$ that were sampled in this phase, but less than $z_v/2 - L/(16K_0\log^6 L)$ of them turned out to be included in $L'_v$, then we freeze node $v$ and all of its unsampled colors $(v, q')$. Moreover, if for a node $v$ and a color $q\in L_v$, in this phase we sampled at least $z_{v, q} \geq L/(16K_0 C \log^6 L)$ of colors $(u, q)$ in neighboring nodes of $v$, but more than $z_{v,q}/2 + L/(16K_0C \log^6 L)$ of them turned out to be included in their respective lists $L'_u$, then we freeze all unsampled colors $(u, q')$, for any $q'$, in neighbors $u$ of $v$. At the end of all the phases, if a node $v$ has less than $L/(2\log^2 L)$ frozen colors $(v, q)$, we discard all of these colors and none of them will be included in $L'_v$.

\begin{lemma}\label[lemma]{lem:shatteringOfListColoring} Each connected component of the graph $H^2$ induced by frozen colors admits, for any integer $\lambda\geq 1$, a $(\lambda, O(\log^{1/\lambda}n \cdot \log^2 \log n))$ network decomposition which can be computed in $\lambda \cdot \log^{1/\lambda} n \cdot 2^{O(\sqrt{\log \log n})}$ rounds, deterministically. 
\end{lemma}
\begin{proof}Follows from \Cref{Shattering} (P3), and the observation that the probability of each color getting frozen is at most $exp(-\tilde{\Omega}(\sqrt{L}))$, and the freezing of colors $(v, q)\in H$ that are more than $5$ hops apart in $H$ depend on disjoint random bits.
\end{proof}

\paragraph{A new LLL for completion of the pruning} Consider the set of frozen colors, and the following new LLL for determining the inclusion of each of these frozen colors in their respective pruned lists, each included with probability $1/2$. We have two bad events: (I) $\mathcal{E}_{v}$ if a node $v$, which has $f_{v} \geq L/(2\log^2 L)$ frozen colors $(v, q)$, less than $f_v/2 - L/(2\log^2 L)$ of these colors get chosen for inclusion in $L'_v$, (II) $\mathcal{E}_{v, q}$ for a node $v$ and a color $q\in L_v$, which has $f_{v, q}$ frozen colors $(u, q)$ in neighboring nodes $u$ of $v$, if more than $f_{v, q}/2 + L/(2C\log^2 L)$ of these frozen colors get chosen for inclusion in their respective pruned lists $L'_u$.

\begin{observation} If we find a fixing for the frozen colors without allowing any of the bad events in the completion LLL to happen, then the overall lists $L'_v$ satisfy the requirements of $2$-factor pruning. Moreover, in this completion LLL, each bad event has probability at most $p=exp(-\tilde{\Omega}(\sqrt{L}))$ and they have dependency $d=O(L^2)$. 
\end{observation}

\begin{lemma} If $L \geq (\log\log n)^4$, then we can solve the completion LLL on each of the connected components of $H^2$ on frozen colors in $2^{O(\sqrt{\log\log n})}$ rounds.
\end{lemma}
\begin{proof} For $L \geq (\log\log n)^4$, the new LLL that we set for completing the pruning satisfies $p(ed)^{\lambda}<1$ for $\lambda=\Omega(\sqrt{\log\log n})$, as it had per-event probability $p=exp(-\tilde{\Omega}(\sqrt{L}))$ and dependency degree $d=O(L^2)$. Hence, we can apply the deterministic algorithm of \Cref{DETLLL}, on top of the network decomposition supplied by \Cref{lem:shatteringOfListColoring} for each of the connected components of $H^2$ on the frozen colors, both with parameter $\lambda=\Omega(\sqrt{\log\log n})$. Thus, we get an algorithm for completing the sampling in $2^{O(\sqrt{\log\log n})}$ rounds.
\end{proof}

\begin{lemma} If $L \leq O(\log^4\log n)$, we can solve the $2$-factor list pruning LLL in $2^{O(\sqrt{\log\log n})}$ rounds.
\end{lemma}
\begin{proof} 
We first devise an algorithm $\mathcal{A}_{n^*}$ which solves the $2$-factor list pruning problem in $O(\log^{1/4} n^*)$ rounds on any $n^*$-node graph, with probability $1-1/n^*$. Then, we use \Cref{lemma:SpecialSpeedup} to speed up this algorithm to solve $n$-node list prunings in $2^{O(\sqrt{\log\log n})}$.

\paragraph{Base algorithm $\mathcal{A}_{n^*}$} As mentioned above, when we target correctness probability $1-1/n^*$, we can without loss of generality assume that $L=O(\log^2 n^*)$. Then, we first perform the $O(\log^4 L)$ rounds of partial sampling of the pruning, as explained above (by going through a defective coloring, and then sampling each of its colors, one by one). We then are left with a number of connected components of the color-choice graph $H^2$ on frozen colors, and new completion LLL for each of them, as described above. Each of these new LLLs satisfies the polynomial LLL criterion $p(ed)^{\lambda}<1$ for $\lambda=\omega(1)$, because it has per-event probability $p=exp(-\tilde{\Omega}(\sqrt{L}))$ and dependency degree $d=O(L^2)$. Hence, we are able to deterministically solve these completion LLLs in $\lambda (\log n^*)^{1/\lambda} 2^{O(\sqrt{\log\log n^*})}$ rounds, using \Cref{DETLLL}. Therefore, the overall complexity is at most $O(\log^4 L) + O(\log^{1/4} n^*) \leq O(\log^{1/4} n^*)$.  

\paragraph{Speed up} Now, we can apply \Cref{lemma:SpecialSpeedup} to speed up this algorithm $\mathcal{A}_{n^*}$. In particular, the procedure of the proof of \Cref{lemma:SpecialSpeedup} will set $n^{*}=\log n$ and then, it will run $A_{n^*}$ on the $n$-node graph, hence forming a new LLL that satisfies a much better exponent of the polynomial LLL criterion, concretely $\lambda=\Omega(\log\log n)$. See \Cref{lemma:SpecialSpeedup} for details. As a result of solving that LLL, we get an algorithm $\mathcal{A}'_n$ that performs the $2$-factor list pruning in $2^{O(\sqrt{\log\log n})}$ rounds, on any $n$-node graph. We note that we have $\leq (\log\log n)^4$ and thus the dependency degree in the pruning LLL is $d=O(L^2) \ll 2^{O(\log^{1/4} \log n)}$, which satisfies the requirement of \Cref{lemma:SpecialSpeedup}.
\end{proof}

\end{document}